\documentclass[twoside]{article}
\input{taesup.sty}
\usepackage[accepted]{aistats2020}
\usepackage{graphicx}
\usepackage{paralist}
\usepackage{dsfont}
\usepackage{bm}

\usepackage{hyperref}       % hyperlinks

\usepackage[round]{natbib}

\begin{document}

%%%%%%%%%%%%%%%%%%%%%%%%%%%%%%%%%%%%%%%%%%%%%%%%%%%%%%%%%%%%%%%%%%%%%%%%%%%
\newenvironment{proof}{\paragraph{Proof:}}{\hfill}%%%%%%%%%%%%%%%%%%%%%%%%%
\newcommand{\ie}{\textit{i}.\textit{e}., }
\newenvironment{proof_thm}{\paragraph{Proof of Theorem:}}{\hfill}%%%%%%%%%%%%%%%%%%%%%%%%%

%%%%%%%%%%%%%%%%%%%%%%%%%%%%%%%%%%%%%%%%%%%%%%%%%%%%%%%%%%%%%%%%%%%%%%%%%%%



% \twocolumn[
% \aistatstitle{Unsupervised Neural Universal Denoiser for \\Finite-Input General-Output Noisy Channel}
% \aistatsauthor{ Tae-Eon Park, Taesup Moon}
% % \And Byunghan Lee \And  Seonwoo Min, Sungroh Yoon}
% \aistatsaddress{Department of Electrical and Computer Engineering, Sungkyunkwan University, Suwon, Korea 16419 \\ \texttt{\{pte1236, tsmoon\}@skku.edu}
% %  \And  Seoul National University of \\ Science and Technology \\ Seoul, Korea 01811 \\ \texttt{bhlee@seoultech.ac.kr}
% %  \And Seoul National University \\ Seoul, Korea 08826 \\ \texttt{\{mswzeus, sryoon\}@snu.ac.kr}
%  } ]

% \twocolumn[
% \aistatstitle{Unsupervised Neural Universal Denoiser for \\Finite-Input General-Output Noisy Channel}
% \aistatsauthor{Anonymous Authors}
% \aistatsaddress{Unknown Institution
%  } ]

\twocolumn[
\aistatstitle{Unsupervised Neural Universal Denoiser for Finite-Input General-Output Noisy Channel}
\aistatsauthor{ Tae-Eon Park and Taesup Moon}
% \And Byunghan Lee \And  Seonwoo Min, Sungroh Yoon}
\aistatsaddress{Department of Electrical and Computer Engineering\\ Sungkyunkwan University (SKKU), Suwon, Korea 16419 \\ \texttt{\{pte1236, tsmoon\}@skku.edu}
%  \And  Seoul National University of \\ Science and Technology \\ Seoul, Korea 01811 \\ \texttt{bhlee@seoultech.ac.kr}
%  \And Seoul National University \\ Seoul, Korea 08826 \\ \texttt{\{mswzeus, sryoon\}@snu.ac.kr}
 } ]

\begin{abstract}
We devise a novel neural network-based universal denoiser for the finite-input, general-output (FIGO) channel. Based on the assumption of known noisy channel densities, which is realistic in many practical scenarios, we train the network such that it can denoise as well as the \emph{best} sliding window denoiser for \emph{any} given underlying clean source data. Our algorithm, dubbed as Generalized CUDE (Gen-CUDE), enjoys several desirable properties; it can be trained in an unsupervised manner (solely based on the noisy observation data), has much smaller computational complexity compared to the previously developed universal denoiser for the same setting, and has much tighter upper bound on the denoising performance, which is obtained by a theoretical analysis. In our experiments, we show such tighter upper bound is also realized in practice by showing that Gen-CUDE achieves much better denoising results compared to other strong baselines for both synthetic and real underlying clean sequences.

% We utilize a single neural network and devise a way of unsupervised training it as a sliding-window denoiser

%   We utilized neural network to calculate Bayes Response which is the argument of minimum achievable expected loss. With neural network, we could achieve efficient and satisfactory-level universal denoising algorithm. To prevent overfitting and to draw good performance of neural network, transformation of mathematical expression and quantization with decision boundaries are needed in application of neural network. In this paper, we show the detailed process and experimental results.
\end{abstract}

%!TEX root = DICO_ndude_AISTAT.tex
\section{Introduction}  
Denoising is a ubiquitous problem that lies at the heart of a wide range of fields such as statistics, engineering, bioinformatics, and machine learning. While numerous approaches have been undertaken, many of them focused on the case for which both the input and output of a noisy channel are continuous-valued \citep{DonJoh95,ElaAha06, bm3d}. In addition, discrete denoising, in which the input and output of the channel take their values in some finite set, have also been considered more recently \citep{Dude, moon2016neural, sdude}.

In this paper, we focus on the hybrid case, namely, the setting in which the underlying clean input source is finite-valued, while the noisy channel output can be continuous-valued. Such scenario naturally occurs in several applications; for example, in DNA sequencing, the finite-valued nucleotides (A,C,G,T) are typically sequenced through observing the continuous-valued light intensities, also known as \emph{flowgrams}. Other examples can be found in digital communication, in which the finite-valued codewords are modulated, e.g., QAM, and sent via a Gaussian channel, as well as in speech recognition, in which the finite-valued phonemes are observed as continuous-valued speech waveforms. In all of above examples, the goal of denoising 	is to recover the underlying finite-valued clean input source from the continuous-valued noisy observations.

There are two standard approaches for tackling above problem: supervised learning and Bayesian learning approaches. The supervised learning collects many clean-noisy paired data and learn a parametric model, e.g., neural networks, that maps noisy to clean data. While simple and straightforward, applying supervised learning often becomes challenging for the applications in which collecting underlying clean data is unrealistic. For such \emph{unsupervised} setting, a common practice is to apply the Bayesian learning framework. That is, assume the existence of stochastic models on the source and noisy channel, then pursue the optimum estimation with respect to the learned joint distribution. Such approach makes sense for the case in which precisely modeling or designing the clean source is possible, \emph{e.g.}, in digital communication, but limitations can also arise when the assumed stochastic model fails to accurately reflect the real data distribution.

% Moreover, obtaining accurate models on the original clean source is trickier than obtaining the noise model. 

% A standard approach to tackle above problem is to assume the existence of stochastic models on the source and noisy data, then pursue the optimum estimation with respect to the assumed or learned joint distribution, i.e., the Bayesian approach. Such approach makes sense for the case in which precisely modeling or designing the clean source is possible, \emph{e.g.}, in digital communication, but limitations can arise when the assumed stochastic model fails to accurately reflect the real data distribution. Moreover, obtaining accurate models on the original clean source is trickier than obtaining the noise model. 

As a third alternative, the so-called \emph{universal} approach has been proposed in \citep{Dude,dembo2005universal}. Namely, while remaining in the unsupervised setting as in the Bayesian learning, the approach makes \emph{no} assumption on the source stochastic model and instead applies the competitive analysis framework; namely, it focuses on the class of \emph{sliding window denoisers} and aims to asymptotically achieve the performance of the best sliding window denoiser for \emph{all} possible sources, solely based on the knowledge on the noisy channel model. The pioneering work, \citep{Dude}, devised Discrete Universal DEnoiser (DUDE) algorithm, which handled the finite-input, finite-output (FIFO) setting, and \citep{dembo2005universal} extended it to the case of finite-input, general-output (FIGO) channels, the setting on which this paper focuses. 

While above both universal schemes enjoyed strong theoretical performance guarantees, they both had critical algorithmic limitations. Namely, the original DUDE becomes very sensitive to the selection of a hyperparameter, i.e., the window size $k$, and the generalized scheme for FIGO channel additionally suffered from the prohibitive computational complexity. Recently, \citep{moon2016neural,ryu2018conditional} employed neural networks in place of a counting vector used in DUDE and showed their schemes can significantly improve the denoising performance and robustness of DUDE. In this paper, we aim to extend the generalized scheme of \citep{dembo2005universal} for the FIGO channel toward the direction of \citep{moon2016neural,ryu2018conditional}, i.e., utilize neural networks to achieve much faster and better performance. Such extension is not straightforward, as we argue in the later sections, due to the critical difference that the channel has the continuous-valued outputs.

% aims to asymptotically achieve the performance of the \emph{best} sliding window denoiser for \emph{all} possible sources. 

% To that end, the universal denoising approach, which does not assume the availability of neither the clean source data nor the source stochastic model, any stochastic model on the clean source, has been proposed in \cite{Dude}. The approach focused on the class of \emph{sliding window denoisers} and aims to asymptotically achieve the performance of the \emph{best} sliding window denoiser for \emph{all} possible sources, solely based on the knowledge on the noisy channel. 

% Moreover, while \cite{Dude} focused on the discrete-input, discrete-output setting,  \cite{dembo2005universal} extended it to the finite-input, general-output (FIGO) channel case and proposed a polynomial-time universal sliding-window denoiser. 

% While \cite{dembo2005universal} presented rigorous theoretical analyses on the universality of the proposed method, the computational complexity of the algorithm became prohibitive as the window size $k$ increases, which makes the algorithm impractical. To that end, in this paper, we propose a practical neural network-based denoiser that achieves excellent denoising performance for the Markov chain as well as the real DNA sequences that have homopolymer errors. 

Our contribution is threefold:
\begin{compactitem}
    \item Algorithmic: We develop a new neural network-based denoising algorithm, dubbed as Generalized CUDE (Gen-CUDE), which can run orders of magnitude faster than the previous state-of-the-art in \citep{dembo2005universal}. 
    \item Theoretical: We give a rigorous theoretical analyses on the performance of our method and obtain a much tighter upper bound on the average loss compared to that of \citep{dembo2005universal}.
    \item Experimental: We compare our algorithm on denoising both the simulated and real source data and show the superb performance compared to other strong baselines.
\end{compactitem}

\section{Notations and Problem Setting}\label{sec:notation}

% \subsection{Notations and Problem Setting}\label{subsec:notation}

We follow \citep{dembo2005universal} but give more succinct notations. Throughout this paper, we will generally denote a sequence ($n$-tuple) as, e.g., $a^n=(a_1,\ldots,a_n)$, and $a_i^j$ referes to the subsequence $(a_i,\ldots,a_j)$. We denote the clean, underlying source data as $x^n$ and assume each component $x_i$ takes a value in some finite set $\mathcal{A}=\{0,\ldots,M-1\}$. The lowercase letters are used to denote the \emph{individual} sequences or the realization of a random sequence. We assume the noisy channel, denoted as $\mathcal{C}$, is \emph{memoryless} and is given by the set $\{f_a\}_{a\in\mcA}$, in which $f_a$ denoting the density with respect to the Lebesgue measure\footnote{We assume such density always exists for concreteness.} associated with the channel output distribution for an input symbol $a$. Following states the mild assumption that we make on $\{f_a\}_{a\in\mcA}$ throughout the paper. 

\begin{assumption}\label{asmt1}
The set of densities $\{f_a\}_{a\in\mcA}$ is a set of linearly independent functions in $L_1(\mu)$.
\end{assumption}
Given above channel $\mathcal{C}$, the noise-corrupted version of the source sequence $x^n$ is denoted as $Y^n=(Y_1,\ldots,Y_n)$. Note we used the uppercase letter to emphasize the randomness in the noisy observation. 
Now, consider a measurable quantizer $Q:\mathbb{R}\rightarrow \mcA$, which quantizes the channel output to symbols in $\mcA$, and the induced channel $\Pib$, a $M\times M$ channel transition matrix induced by $Q$ and $\{f_a\}$. We denote the quantized output of $Y^n$ by $Z^n$, and the $(x,z)$-th element of $\Pib$ can be computed as
\be
\Pib(x,z)=\int_{y:Q(y)=z}f_x(y)dy.
\label{eq:ind_pi}
\ee
Note Assumption \ref{asmt1} ensures that $\Pib$ is an invertible matrix. Moreover, we denote $Z_i\triangleq Q(Y_i)$ as the quantized version of $Y_i$.
% The quantized channel output of $Y^n$ will be denoted as $Z^n$. 

Given the entire \emph{continuous-valued} noisy observation $Y^n$, the denoiser reconstructs the original \emph{discrete} input $x^n$ with $\hat{X}^n=(\hat{X}_1(Y^n),\ldots,\hat{X}_n(Y^n))$, where each reconstructed symbol $\hat{X}_i(Y^n)$ takes its value in $\mcA$. The fidelity of the denoising is measured by the average loss
\be
L(x^n,\Xhat^n(Y^n))=\frac{1}{n}\sum_{i=1}^n\Lb\Big(x_i,\Xhat_i(Y^n)\Big),\label{eq:avg-loss}
\ee
in which $\Lb\in\mathbb{R}^{M\times M}$ is a per-symbol bounded loss matrix. 
Moreover, we denote $\Lb_{\hat{x}}$ as the $\hat{x}$-th column of $\Lb$ and $\Lambda_{\max}=\max_{x,\hat{x}}\mathbf{\Lambda}(x,\hat{x})$.
% and $\Delta^M$ as a space of simplex vectors in $\mathbb{R}^M$. 
% \subsubsection{Bayes Envelope and Response}

Then, for a probability vector $\Pb\in\Delta^M$, the \emph{Bayes envelope}, $U(\Pb)$, is defined to be 
% Bayes Envelope $U(P_X)$ is a minimum achievable expected loss when guessing the value of a variable $X$ distributed according to $P_X$ 
%$$U(P_X):= \min_{\hat{x}} \mathbb{E}[{\Lambda(X,\hat{x})}]  \qquad \qquad ... (1) $$
\begin{equation}
   U(\Pb)\triangleq \min_{\hat{x}\in\mcA} \sum_{x\in\mcA} \Lb(x,\hat{x})\Pb(x),
%   = \min_{\hat{x}\in\mcA}\bm\Lambda_{\hat{x}}^\top\Pb,
   \label{eq:bayes_e}
    % \mathbb{E}[{\Lambda(X,\hat{x})}] \tag{1}
\end{equation}
which, in words, stands for the minimum achievable expected loss in estimating the source symbol that is distributed according to $\Pb$. The argument that achieves (\ref{eq:bayes_e}) is denoted by $\mathcal{B}(\mathbf{P})$, the \emph{Bayes response} with respect to $\mathbf{P}$. Furthermore, in the later sections, we extend the notion of Bayes response by using $\mathbf{P}$ that is not necessarily a probability vector.  

The $k$-th order sliding-window denoisers are the denoisers that are defined by the time-invariant mappings $g_k:\mathbb{R}^{2k+1}\rightarrow \mcA$. That is, $\hat{X}_i(Y^n)=g_k(Y_{i-k}^{i+k})$. We also denote the tuple $ \mathbf{Y}^{(k)}_{-i}\triangleq(Y_{i-k}^{i-1},Y_{i+1}^{i+k})$ as the $k$-th order context around the noisy symbol $Y_i$.

% We utilize the notion of Bayes response in devising our algorithm although the underlying clean source may not have any probability distribution. 

% this notation for devising our universal denoiser although the underlying clean source does not have any probability distribution. 

% by choosing $\hat{x}$ where the expectation 
% \subsubsection{Bayes Response}

%  Bayes Response is the point x for which expected loss attains its smallest value. (i.e. the argument of minimum) 
%  %$$\hat{X}_{Bayes}(P_X):= \arg\min_{\hat{x}}\mathbb{E}[{\Lambda(X,\hat{x})}]\qquad \qquad ... (2)$$ 
%  \begin{equation}
%     \hat{X}_{Bayes}(P_X):= \arg\min_{\hat{x}}\mathbb{E}[{\Lambda(X,\hat{x})}]\tag{2}
% \end{equation}

\section{Related Work}
% \subsubsection{DUDE and Generalized DUDE}

\subsection{DUDE, Neural DUDE and CUDE}\label{subsec:dude}

A straightforward baseline for the FIGO channel setting is to simply quantize the continuous-valued output and apply the discrete denoising algorithm to estimate the underlying clean source. While such scheme is clearly suboptimal since it significantly discards the information observed in $Y^n$, we briefly review the previous work on discrete denoising so that we can build intuitions for devising our algorithm for the FIGO channel.

% As mentioned in the Introduction, \citep{weissman2005universal} first considered the discrete universal denoising algorithm, DUDE. 

\noindent \textbf{DUDE} was devised by \citep{Dude} and is a two-pass, sliding-window denoiser for the FIFO setting. In discrete denoising, we denote $Z^n$ as the finite-valued noisy sequence, $\mathbf{Z}_{-i}^{(k)}\triangleq(Z_{i-k}^{i-1},Z_{i+1}^{i+k})$ as the $k$-th order context around $Z_i$, and $\mathbf{\Gamma}$ as the Discrete Memoryless Channel (DMC) transition matrix that induces the noisy sequence $Z^n$ from the clean $x^n$. Then, the reconstruction of DUDE at location $i$ is defined to be
\begin{eqnarray}
\hat{X}_{i}(\mathbf{Z}_{-i}^{(k)},Z_i)= \arg\min_{\hat{x}\in\hat{\Xcal}}\hat{\mathbf{p}}_{\text{emp}}(\cdot|\mathbf{Z}_{-i}^{(k)})^\top\mathbf{\Gamma}^{\dagger}[\boldsymbol\Lambda_{\hat{x}}\odot\bm\gamma_{Z_i}],\label{eq:dude_rule}
\end{eqnarray}
in which $\mathbf{\Gamma}^{\dagger}$ is a Moore-Penrose pseudo-inverse of $\mathbf{\Gamma}$ (assuming $\mathbf{\Gamma}$ is full row-rank), $\bm\gamma_z$ is the $z$-th column of $\mathbf{\Gamma}$, and $\hat{\mathbf{p}}_{\text{emp}}(\cdot|\mathbf{Z}_{-i}^{(k)})\in\mathbb{R}^{|\mcZ|}$ is an empirical probability vector on $Z_i$ given the context $\mathbf{Z}_{-i}^{(k)}$, obtained from the entire noisy sequence $Z^n$. That is, for a $k$-th order double-sided context $\mathbf{Z}^{(k)}$, the $z$-th element of $\hat{\mathbf{p}}_{\text{emp}}(\cdot|\mathbf{Z}_{-i}^{(k)})$ becomes
\be
\hat{\mathbf{p}}_{\text{emp}}(z|\mathbf{Z}^{(k)})=\frac{|\{j:\mathbf{Z}_{-j}^{(k)}=\mathbf{Z}^{(k)}, Z_j=z\}|}{|\{j:\mathbf{Z}_{-j}^{(k)}=\mathbf{Z}^{(k)}\}|}.\label{eq:emp}
\ee
% Moreover, the $\boldsymbol\Lambda_{\hat{x}}$ and $\bm\pi_{Z_i}$ in (\ref{eq:dude_rule}) stand for the $\hat{x}$-th and $Z_i$-th column of $\Lb$ and $\Pib$, respectively. Note the rule (\ref{eq:dude_rule}) solely depends on $Z^n$ and the knowledge of $\Pib$ is required. 
The main intuition for obtaining (\ref{eq:dude_rule}) is to show that the true posterior distribution can be approximated by using (\ref{eq:emp}) and inverting the DMC channel, $\mathbf{\Gamma}$. That is, the following approximation 
\be
p(x_i|Z_{i-k}^{i+k})\approx \big(\bm\gamma_{Z_i}\odot [\mathbf{\Gamma}^{\dagger\top}\hat{\mathbf{p}}_{\text{emp}}(\cdot|\mathbf{Z}_{-i}^{(k)})]\big)_{x_i} \label{eq:approx_emp}
\ee
holds with high probability with large $n$ \citep[Section IV.B]{Dude}. Then, for each location $i$, (\ref{eq:dude_rule}) is the $\mathcal{B}(\bm\gamma_{Z_i}\odot [\mathbf{\Gamma}^{\dagger\top}\hat{\mathbf{p}}_{\text{emp}}(\cdot|\mathbf{Z}_{-i}^{(k)})])$, the Bayes response with respect to the right-hand side of (\ref{eq:approx_emp}).
 % $\bm\pi_{Z_i}\odot [\Pib^{\dagger\top}\hat{\mathbf{p}}_{\text{emp},Z}(\cdot|\Cb_i)]$.
% the approximation vector. {}
% identify that for stationary clear source $x^n$,
% \be
% \hat{X}_{i,\text{opt}}(\Cb_i,Z_i) = \arg\min_{\hat{x}}\sum_{x}\Lb(x,\hat{x})p(x|\Cb_i,Z_i)
% \ee
% and show the following approximation is accurate enough:
% \be
% p(x|\Cb_i,Z_i)\approx \big(\bm\pi_{Z_i}\odot [\Pib^{\dagger\top}\hat{\mathbf{p}}_{\text{emp},Z}(\cdot|\Cb_i)]\big)_x
% \ee
% The rule (\ref{eq:dude_rule}) is based on the intuition from achieving the Bayes response using (\ref{eq:emp}) and inverting the channel. 
\citep{Dude} showed the DUDE rule, (\ref{eq:dude_rule}), can universally attain the denoising performance of the best $k$-th order sliding window denoiser for \emph{any} $x^n$. \vspace{.05in}
% For more rigorous details and results, we refer to  \citep{Dude}.\vspace{.05in}

% \citep{Dude} showed the \emph{universality} of DUDE in the sense that it can attain the performance of the best sliding-window denoiser for any underlying sequence $x^n$, provided that $k$ grows appropriately with $n$.
% for any underlying stationary process, it asymptotically attains the Bayes optimal preformance provided that $k$ grows appropriately with $n$. 
% For more details, we refer to the original paper \citep{Dude}.

% The main gist of \citep{moon2016neural} was to devise the ``pseudo-labels'' based on the unbiased estimated loss matrix $\Ellb$ and use them as target labels when training a DNN-based sliding window denoiser. 

\noindent\textbf{Neural DUDE (N-DUDE)} was recently proposed by \citep{moon2016neural}, and it
% takes the alternative view on the sliding-window denoiser mentioned above and 
identified that the limitation of DUDE follows from the empirical count step in (\ref{eq:emp}). Namely, the count happens totally separately for each context $\Cb$, even if the contexts can be very similar to each other. To that end, N-DUDE implements a \emph{single} neural network-based sliding-window denoiser such that the information among similar contexts can be shared through the network parameters. That is, N-DUDE defines $\mathbf{p}^k_{\texttt{N-DUDE}}(\mathbf{w},\cdot):\mathcal{Z}^{2k}\rightarrow\Delta^{|\mcS|}$, in which $\wb$ stands for the parameters in the network, and $\mcS$ is a set of single-symbol denoisers, $s:\mcZ\rightarrow\mcA$, which map $\mcZ$ to $\mcA$. Thus, $\mathbf{p}^k_{\texttt{N-DUDE}}(\mathbf{w},\cdot)$ takes the context $\mathbf{Z}_{-i}^{(k)}$ and outputs a probability distribution on the single-symbol denoisers to apply to $Z_i$, for each $i$. Note for discrete denoising, $|\mcS|$ has to be finite, hence the network has the structure of a multi-class classification network.

To train the network parameters $\wb$, N-DUDE defines the objective function
% Then, the objective function of N-DUDE to train $\wb$ becomes 
\begin{align}
\mathcal{L}(\mathbf{w}, Z^n)\triangleq&\frac{1}{n}\sum_{i=1}^n\mathbb{C.E}\Big(\mathbf{L}_{\text{new}}^\top\mathds{1}_{Z_i}, \mathbf{p}^k_{\texttt{N-DUDE}}(\mathbf{w},\mathbf{C}_i)\Big),\nonumber
% \label{eq:objective}
\end{align}
in which $\mathbb{C.E}(\mathbf{g},\mathbf{p})$ stands for the (unnormalized) cross-entropy, and $\mathbf{L}_{\text{new}}^\top\mathds{1}_{Z_i}$ is the \emph{pseudo-label} vector for the $i$-th location, calculated from the unbiased estimate of the true expected loss which can be computed with $\Lb$, $\mathbf{\Gamma}$, and $Z^n$ (more details are in \citep{moon2016neural}). Note the dependency of the objective function on $Z^n$ is highlighted, hence, the training of $\wb$ is done in an unsupervised manner together with the knowledge of the channel.

Once the objective function is minimized via stochastic gradient descent, the converged parameter is denoted as $\wb^\star$. 
% Then, for each context $\Cb\in\Cb[k]$, the N-DUDE outputs a probability distribution on the single-letter mappings and chooses 
Then, the single-letter mapping defined by N-DUDE for the context $\mathbf{Z}_{-i}^{(k)}$ is expressed as 
% \begin{align}
$s_{k,\texttt{N-DUDE}}(\mathbf{Z}_{-i}^{(k)},\cdot)=\arg\max_{s\in\mcS}\mathbf{p}^k_{\texttt{N-DUDE}}(\wb^\star,\mathbf{Z}_{-i}^{(k)})_s,$
% \end{align}
% \begin{align}
% s_{k,\text{N-DUDE}}(\Cb,\cdot)=\arg\max_{s\in\mcS}\mathbf{p}^k(\wb^\star,\Cb)_s,\label{eq:n_dude dfn}
% \end{align}
% Note $\mathbf{p}^k(\wb^*,\cb)_s$ in (\ref{eq:n_dude dfn}) stands for the $s$-th element of the probability vector $\mathbf{p}^k(\wb^*,\cdot)\in\Delta^{|\mcS|}$. 
and the reconstruction at location $i$ becomes 
\be
\hat{X}_{i,\texttt{N-DUDE}}(\mathbf{Z}_{-i}^{(k)},Z_i)=s_{k,\texttt{N-DUDE}}(\mathbf{Z}_{-i}^{(k)},Z_i).\label{eq:n_dude_recons}
\ee
\citep{moon2016neural} shows N-DUDE significantly outperforms DUDE and is more robust with respect to $k$. \vspace{-.1in}

% Neural DUDE denoises the noisy data after adaptively training the network parameters with the \emph{same} noisy data. \citep{moon2016neural} shows N-DUDE significantly outperforms DUDE, has more robustness with respect to $k$, and gets very close to the \emph{optimum} denoising performance for stationary sources. 
% \citep{moon2016neural} shows encouraing empirical results including the performance of N-DUDE being robust with respect to $k$.
% , and in this paper, we provide theoretical justification of Neural DUDE.

\noindent\textbf{CUDE} was proposed by \citep{ryu2018conditional} following-up on N-DUDE, which took an alternative and simpler approach of using neural network to extend DUDE. Namely, instead of using the empirical distribution in (\ref{eq:emp}), CUDE learns a network $\mathbf{p}^k_{\texttt{CUDE}}(\wb,\cdot):\mcZ^{2k}\rightarrow\Delta^{|\mcZ|}$, which takes the context $\mathbf{Z}_{-i}^{(k)}$ as input and outputs a prediction for $Z_i$, by minimizing $\frac{1}{n}\sum_{i=1}^n\mathbb{C.E}(\mathds{1}_{Z_i}, \mathbf{p}^k_{\texttt{CUDE}}(\wb,\mathbf{Z}_{-i}^{(k)})).$ Thus, the network aims to directly learn the conditional distribution of $Z_i$ given its context $\mathbf{Z}_{-i}^{(k)}$.
% borrows the idea of using neural network from N-DUDE, but uses it differently to replace  (\ref{eq:emp}) in the DUDE rule with a neural network learned empirical distribution of $z$ given $\Cb$. Namely, it defines a network $\mathbf{p}_{\text{emp}}(\wb,\cdot):\mcZ^{2k}\rightarrow\Delta^{|\mcZ|}$ and train it by minimizing
% $
% \frac{1}{n}\sum_{i=1}^n\mathcal{C}(\mathds{1}_{Z_i}, \mathbf{p}_{\text{emp}}(\wb,\Cb_i)).
% $
Once the minimizer $\wb^*$ is obtained, CUDE then simply plugs in $\mathbf{p}^k_{\texttt{CUDE}}(\wb^*,\mathbf{Z}_{-i}^{(k)})$ in place of $\hat{\mathbf{p}}_{\text{emp}}(\cdot|\mathbf{Z}_{-i}^{(k)})$ in  (\ref{eq:dude_rule}). \citep{ryu2018conditional} shows that CUDE outperforms N-DUDE primarily due to the reduced output size of the neural network, \ie $|\mcZ|$ vs. $|\mcS|=|\mcA|^{|\mcZ|}$.\vspace{.1in}

% The universal denoising problem for discrete-input, discrete-output channel originally has been considered in \citep{weissman2005universal}.
%and \citep{ryu2018conditional,moon2016neural}. 
% 듀드 수식, mvector, ...›
% To be rigorously, DUDE carries out single symbol denoising based on the following formula:
% \begin{align*}
%     \hat{X}_{i,DUDE}&(z_{i-k}^{i-1},Z_{i},z_{i+1}^{i+k}) \tag{4} \label{eq:DUDE} \\
%      &=\arg \min _{\hat{x} \in \mathcal{A}} \mathbf{m}^{T}(z^{n},z_{i-k}^{i-1}, z_{i+1}^{i+k}) \mathbf{\Gamma}^{-1}\left[\boldsymbol{\lambda}_{\hat{x}} \odot \boldsymbol{\pi}_{Z_{i}}\right].
% \end{align*}
% Here, $\mathbf{m}\left(a^{n}, b^{k}, c^{k}\right)$ denote the $M$-dimensional empirical count vector with $2k < n, a^{n} \in \mathcal{A}^{n}, b^{k} \in \mathcal{A}^{k}, c^{k} \in \mathcal{A}^{k} $:
% \begin{align*}
%     \mathbf{m}&\left(a^{n}, b^{k}, c^{k}\right)[\beta] \tag{5} \label{eq:mvector1}\\
%     &=\left|\left\{i : k+1 \leq i \leq n-k, a_{i-k}^{i+k}=b^{k} \beta c^{k}\right\}\right|.
% \end{align*}

\vspace{-.1in}

\subsection{Generalized DUDE for FIGO channel}

\citep{dembo2005universal} extended DUDE algorithm specifically for the FIGO channel case, and we refer to their scheme as Generalized DUDE (Gen-DUDE) from now on. The key challenge arises in the FIGO channel for applying the DUDE framework is that it becomes impossible to obtain an empirical distribution like (\ref{eq:emp}) based on counting for each context, because there are infinitely many possible contexts. 

Therefore, by denoting $\mathbf{P}(X_i | y_{i-k}^{i+k})\in\Delta^{M}$ as the conditional probability vector on $X_i$ given the $(2k+1)$-tuple $y_{i-k}^{i+k}$, Gen-DUDE first identifies that the denoising rule at location $i$ should be the Bayes response
\begin{align}
\hat{X}_i(y^n) =& \ \mathcal{B}(\mathbf{P}(X_i\mid y_{i-k}^{i+k}))= \ \mathcal{B}(\mathbf{P}(X_i,y_{i-k}^{i+k})), \label{eq:g-dude}
\end{align}
in which the second equality in (\ref{eq:g-dude}) follows from ignoring the normalization factor of $\mathbf{P}(X_i | y_{i-k}^{i+k})\in\Delta^{M}$. Note, $\mathbf{P}(X_i,y_{i-k}^{i+k})$ is not a probability vector, but the notion of Bayes response still holds. Now, the joint distribution can be expanded as 
\begin{align}
&p(X_i=a,y_{i-k}^{i+k}) =\sum_{u_{-k}^k:u_0=a} p(X_{i-k}^{i+k}=u_{-k}^k,y_{i-k}^{i+k})\nonumber\\
                       =& \sum_{u_{-k}^k:u_0=a} \underbrace{\Big[\prod_{j=-k}^k f_{u_{j}}(y_{i+j})\Big]}_{(a)} \underbrace{p(X_{i-k}^{i+k}=u_{-k}^k)}_{(b)}\label{eq:joint},
\end{align}
in which term (a) of (\ref{eq:joint}) follows from the memoryless assumption on the channel $\mathcal{C}$. Then, Gen-DUDE approximates term (b) of (\ref{eq:joint}), which is now the distribution on the finite-valued source $(2k+1)$-tuples, by computing the empirical distribution of the \emph{quantized} noisy sequence $Z^n$ and inverting the induced DMC matrix $\Pib$, both of which are defined in Section \ref{sec:notation}. Once the approximation for (\ref{eq:joint}) is done, the Gen-DUDE simply computes the Bayes response as in (\ref{eq:g-dude}) with the approximate joint probability vector. For more details, we refer to the paper \citep{dembo2005universal}.

The main critical drawback of Gen-DUDE is the computation required for approximating (\ref{eq:joint}). Namely, the summation in (\ref{eq:joint}) is over all possible $(2k+1)$-tuples of the source symbols, of which complexity grows exponentially with $k$. Therefore, the running time of the algorithm becomes totally impractical even for the modest alphabet sizes, e.g., 4 or 10, as shown in our experimental results in the later section. Moreover, such exponential dependency on $k$ also appears in the theoretical analyses of Gen-DUDE. That is, it is shown that the upper bound on the probability that the average loss of Gen-DUDE deviates from that of the best sliding-window denoiser is proportional to the doubly exponential term $C^{M^{2k+1}}$, which quickly becomes meaningless for, again, modest size of $M$ and $k$. Motivated by such limiations, we introduce neural networks to efficiently approximate $\mathbf{P}_{X_i,y_{i-k}^{i+k}}$ and compute the Bayes response to significantly improve the Gen-DUDE method.

%!TEX root = DICO_ndude.tex

% \section{Utilization of Neural Network}
% \label{headings}
% To apply neural network, we had to decide which part should we leave it to neural network.
% We selected two candidates for target of neural network. Two candidates are as follow :

% \begin{itemize}
% \item  $[P \otimes C]_{U_{0} \mid y_{-k}^{k}} $ from Equation \eqref{eqn3}

% \item  $P(U_{-k}^{k} = u_{-k}^{k})$  from  Equation \eqref{eqn4}

% \end{itemize}
% But here, even if the neural network found the second candidate well, there were still some problems.
% That problem was an enormous amount of calculations.
% we have to sum at least $\mid  X \mid ^{2k} $ times because from second line of $g_{opt}[P](y_{-k}^{k})$, $\sum_{u_{-k}^{k} \in A^{2k+1} : u_{0}=a}$ was there.
% When window size k is large enough, denoising with Generalized DUDE was very slow or impossible due to Memory Error of machine.
% By consideration of above two candidates, we aimed the first candidate as target of neural network.
% So, our target is $[P \otimes C]_{U_{0} \mid y_{-k}^{k}} $ from now on.

% \vspace{1\baselineskip}
\vspace{-.05in}
\section{Main Results}
% \vspace{-.05in}

\subsection{Intuition for Gen-CUDE}
% \vspace{-.05in}

As mentioned above, the Gen-DUDE suffers from high computational complexity due to the expansion given in (\ref{eq:joint}) that requires the summation over the exponentially many (in $k$) terms. The main reason for such expansion in \citep{dembo2005universal} was to utilize the tools of DUDE for approximating term (b) in (\ref{eq:joint}), which inevitably requires to enumerate all the $2k$-tuple terms. Hence, we instead try to directly approximate $\mathbf{P}(X_i|y_{i-k}^{i+k})$ using a neural network. 

Our algorithm is inspired by N-DUDE and CUDE, mentioned in Section \ref{subsec:dude}, which show much better traits compared to the original DUDE. However, we easily notice that the approach of N-DUDE cannot be applied to the FIGO channel case, because there will be infinitely many single-symbol denoisers $s:\mathcal{Y}\rightarrow \mcA$. 
Hence, the output layer of the network should perform some sort of regression, instead of the classification as in N-DUDE, but obtaining the pseudo-label for training in that case is far from being straightforward. Therefore, we take an inspiration from CUDE and develop our Generalized CUDE (Gen-CUDE). 
% CUDE and devise our algorithm as follows. 

% For our algorithm, we consider the quantized noisy sequence $Z^n$ and the induced DMC channel. 

In order to build the core intuition for our method, first consider the quantized noisy sequence $Z^n$ and the induced DMC matrix $\Pib$ (defined in (\ref{eq:ind_pi})). 
That is, $Z_i=Q(Y_i)$ where $Q(\cdot)$ is the quantizer introduced in Section \ref{sec:notation}. 
Furthermore, denote $\mathbf{P}(X_0|y_{-k}^k)\in\Delta^M$ and $\mathbf{P}(Z_0|y_{-k}^k)\in\Delta^M$ as the conditional probability vectors of $X_0$ and $Z_0$ given a $(2k+1)$ tuple $y_{-k}^k$ that appear in the noisy observation $Y^n$. Also, let $\bm f_{X_0}(y_0)\in\mathbb{R}^M$ be the vector of density values of which $a$-th element is $f_{a}(y_0)$. We treat all the vectors as \emph{row} vectors. The following lemma builds the key motivation.

\begin{lemma} Given $y_{-k}^k$, the following holds.
\begin{equation}
\mathbf{P}(X_0|y_{-k}^k)\propto [\mathbf{P}(Z_0|\mathbf{y}_{-0}^{(k)}) \cdot \Pib^{-1}]\odot\bm f_{X_0}(y_0)\label{eq:lemma}
\end{equation}
Namely, we can compute $\mathbf{P}(X_0|y_{-k}^k)$ up to a normalization constant using the conditional distribution on $Z_0$, and the information on the channel $\mathcal{C}$.
\label{lemma_1}
\end{lemma}
% \emph{Proof:} 
\vspace{-.1in}
\emph{Proof:} We have the following chain of equalities for the conditional distribution $p(x_{0}|y_{-k}^{k})$:
\begin{align}
     p(x_{0}|y_{-k}^{k}) =& \frac{p(x_{0} , y_{-k}^{k})}{p(y_{-k}^{k})}  =\frac{p(x_{0} , \mathbf{y}_{-0}^{(k)})f_{x_0}(y_0)}{p(y_{-k}^{k})}\label{eq:equality1}\\
     =& p(x_{0} | \mathbf{y}_{-0}^{(k)})f_{x_0}(y_0)\cdot \frac{p( \mathbf{y}_{-0}^{(k)})}{p(y_{-k}^{k})},\label{eq:equality2}
\end{align}
in which the second equality of (\ref{eq:equality1}) follows from the memoryless property of the densities $\{f_a\}_{a\in\mcA}$ of $\mathcal{C}$, and $\mathbf{y}_{-0}^{(k)}$ stands for the double-sided context $(y_{-k}^{-1},y_1^k)$. Now, by denoting $z_0=Q(y_0)$ as the quantized version of $y_0$, we have the following relation. 
\begin{align}
    p(z_0|\mathbf{y}_{-0}^{(k)})=& \sum_{x_0}p(z_0|x_0, \mathbf{y}_{-0}^{(k)})p(x_0|\mathbf{y}_{-0}^{(k)})\nonumber\\
    =& \sum_{x_0}\Pib(x_0,z_0)p(x_0|\mathbf{y}_{-0}^{(k)}),\label{eq:memless_pi}
\end{align}
in which (\ref{eq:memless_pi}) follows from the channel $\mathcal{C}$ being memoryless and utilizing the notation of the induced DMC matrix, $\Pib$, defined in (\ref{eq:memless_pi}). Thus, following the row vector notations as mentioned above, we have
\begin{align}
    \mathbf{P}(Z_0|\mathbf{y}_{-0}^{(k)})=\mathbf{P}(X_0|\mathbf{y}_{-0}^{(k)})\cdot\Pib.\label{eq:identity}
\end{align}
By inverting $\Pib$ in (\ref{eq:identity}), and combining with (\ref{eq:equality2}) and dropping the term $\frac{p(\mathbf{y}_{-0}^{(k)})}{p(y_{-k}^{k})}$, we have the lemma. \qed
% Thus, we took an inspiration from CUDE and devise our algorithm as below. 

From the lemma, we can see that once we have accurate approximation of the conditional distribution $\mathbf{P}(Z_0|\mathbf{y}_{-0}^{(k)})$, then we can apply (\ref{eq:lemma}) and obtain
the Bayes response with respect to $[\mathbf{P}(Z_0|\mathbf{y}_{-0}^{(k)}) \cdot \Pib^{-1}]\odot\bm f_{X_0}(y_0)$. Now, following the spirit of CUDE, we utilize neural network to approximate $\mathbf{P}(Z_0|\mathbf{y}_{-0}^{(k)})$ from the observed data. We concretely describe our Gen-CUDE algorithm in the next subsection.

% Again, since there are infinitely many possible noisy context, we cannot approximate the distribution with counting, but utilize neural network as CUDE. 

% $\mathcal{B}([\mathbf{P}(Z_0|y_{-k}^{-1},y_1^{k}) \cdot \Pib^{-1}]\odot\bm f_{X_0}(y_0))$

\subsection{Algorithm Description}
Inspired by (\ref{eq:lemma}) and CUDE, we try to use a single neural network to learn the $k$-th order sliding window denoiser. First of all, define $\mathbf{p}^{k}(\mathbf{w},\cdot) : \mathbb{R}^{2k} \rightarrow \Delta^M$ as a feed-forward neural network we utilize. With weight parameter $\mathbf{w}$, the network takes context $\mathbf{y}_{-i}^{(k)}$ as input and send out $\mathbf{P}(Z_0|\mathbf{y}_{-i}^{(k)})$ as output. To learn the parameter $\mathbf{w}$, we define the objective function as 
\begin{align}
\mathcal{L}_{\texttt{Gen-CUDE}}(\mathbf{w}, Y^n)\triangleq&\frac{1}{n-2k}\sum_{i=k}^{n-k}\mathbb{C.E}\Big(\mathds{1}_{Z_i}, \mathbf{p}^k(\mathbf{w},\mathbf{Y}_{-i}^{(k)})\Big).\nonumber
\label{eq:objective}
\end{align}
% L_Gen_CUDE 
% camera-ready -1
Namely, minimizing $\mathcal{L}_{\texttt{Gen-CUDE}}$ leads to training the network to predict the \emph{quantized} middle symbol $Z_i$ based on the continuous-valued context $\mathbf{Y}_{-i}^{(k)}$, hence, the network can maintain the multi-class classification structure with the ordinary softmax output layer. The minimization is done by the stochastic gradient descent-based optimization methods such as Adam \citep{kingma2014adam}. Once the minimization is done, we denote the converged weight vector as $\mathbf{w}^{*}$. Then, by motivated by Lemma \ref{lemma_1}, we define our \texttt{Gen-CUDE} denoiser as the Bayes response with respect to $[\mathbf{p}^k(\mathbf{w}^*,\mathbf{Y}_{-i}^{(k)}) \cdot \Pib^{-1}]\odot \mathbf{f}_{X_0}(Y_i)$ for each $i$. Following summarizes our algorithm.

% Through the above objective, we intended to draw out a mapping function from noisy context to quantized middle symbol. Note that, the middle symbol is ruled out from the context. We experimentally observed that $\mathbf{P}(Z_0|y_{-k}^k)$ solely depends on $y_0$ when the neural network is designed to take context including the middle symbol as input straightforwardly. To deviate from this, and take a step forward in the right direction, we came up with $\mathcal{L}_{\texttt{Gen-CUDE}}$.
%

% By minimizing $\mathcal{L}_{\texttt{Gen-CUDE}}$, we can achieve the effect of learning $\mathbf{w}$. Stochastic Gradient Descent-based optimization method such as Adam \citep{kingma2014adam} can be used for learning $\mathbf{w}$. For the network architecture, we add a softmax layer at the top of fully connected layers. Let $\mathbf{w}^{*}$ denotes a converged weight vectors of the neural network after sufficient weight updates. By calculating the Bayes response with respect to $[\mathbf{p}^k(\mathbf{w}^*,\mathbf{y}_{-i}^{(k)}) \cdot \Pib^{-1}]\odot \mathbf{f}_{X_0}(y_i)$ for each $i$, we implement the \texttt{Gen-CUDE} denoiser.

% However, due to the negative terms in $\Pib^{-1}$, it is possible that some components of $[\mathbf{p}^k(\mathbf{w}^*,\mathbf{Y}_{-i}^{(k)}) \cdot \Pib^{-1}]$ being negative. To estimate $\mathbf{P}(X_0|\mathbf{Y}_{-i}^{(k)})$ in (\ref{eq:identity}) without negative components, we additionally draw out a variation of the network architecture.

\renewcommand{\algorithmicrequire}{\textbf{Input:}}
\renewcommand{\algorithmicensure}{\textbf{Output:}}
\begin{algorithm}
\caption{Gen-CUDE algorithm}
\begin{algorithmic}\label{alg1}
	\REQUIRE Noisy sequence $Y^n$, Context size $k$, $\mathcal{C}=\{f_a\}_{a\in\mcA}$, $\mathbf{\Lambda}$, Quantizer $Q(\cdot)$
    \ENSURE Denoised sequence $\hat{X}_{\text{NN}}^n=\{\hat{X}_{i,\text{NN}}(Y^n)\}_{i=1}^n$
    \STATE Obtain the quantized sequence $Z^n$ using $Q(\cdot)$
    \STATE Compute $\mathbf{\Pi}$ as (\ref{eq:ind_pi}) and initialize $\mathbf{p}^{k}(\mathbf{w},\cdot)$
    %\STATE Compute $\{z_{i}\}_{i=1}^n=\{Q(y_{i})\}_{i=1}^n$
    % \STATE Initialize $\mathbf{p}^{k}(\mathbf{w},\cdot)$
    \STATE Obtain $\mathbf{w}^*$ minimizing $\mathcal{L}_{\texttt{Gen-CUDE}}(\mathbf{w}, Y^n)$ 
    \IF{$i=k+1,\ldots, n-k$}
        % \STATE Compute $\mathbf{p}^{k}(\mathbf{w}^*, (y_{i-k}^{i-1},y_{i+1}^{i+k}))$ and $\mathbf{f}_{X_i}(y_i)$
        \STATE Compute $[\mathbf{p}^{k}(\mathbf{w}^*, \mathbf{Y}_{-i}^{(k)}) \cdot \mathbf{\Pi}^{-1}] \odot \bm f_{X_{i}}(Y_{i})$
        \STATE $\hat{X}_{i,\text{NN}}(y^n)= \mathcal{B}\big([\mathbf{p}^{k}(\mathbf{w}^*, \mathbf{Y}_{-i}^{(k)}) \cdot \mathbf{\Pi}^{-1}] \odot \bm f_{X_{i}}(Y_{i})\big)$
    \ELSE%{$\hat{X}_{i,\text{NN}}(Y^n)=Z_i$}
         \STATE $\hat{X}_{i,\text{NN}}(Y^n)=Z_i$
    \ENDIF
    \STATE Obtain $\hat{X}_{\text{NN}}^n(Y^n)=\{\hat{X}_{i,\text{NN}}(Y^n)\}_{i=1}^n$
\end{algorithmic}
\end{algorithm}

% %%%%%%%%%%%%%%%%%%%%%%%%%%%%%%%%%%%%%%%%%%%%%%%%%%%%%%%%%%%%%%%%%%%%%%%%%%%%%%
% \begin{figure*}[!t]
% \center
% \includegraphics[width=1.75\columnwidth]{AISTATS2020/figs/fig1_merged.png}
% \label{Architecture2}
% \caption{Schematic diagram for architecture with indirect learning. Window size $k=3$ and alphabet size $M=4$ are assumed.}
% \label{fig1}
% \end{figure*}
% %%%%%%%%%%%%%%%%%%%%%%%%%%%%%%%%%%%%%%%%%%%%%%%%%%%%%%%%%%%%%%%%%%%%%%%%%%%%%%

% \subsection{Calculation with Neural Network Output}

%\vspace{1\baselineskip}

\vspace{-.05in}

\subsection{Theoretical Analysis}
In this subsection, we give a theoretical analysis on Gen-CUDE, which follows similar steps as in \citep{dembo2005universal} but derives a much tighter upper bound on the average loss of Gen-CUDE. As a performance target for the competitive analysis, we define the minimum expected loss of $x^{n}$ for the $k$th-order sliding-window denoiser by \vspace{-.05in}
\begin{align}
     D_{x^n}^k = \min _{g_k} \mathbb{E}\Big[\frac{1}{n-2 k} \sum_{i=k+1}^{n-k} \Lb\Big(x_{i}, g_k\Big(Y_{i-k}^{i+k}\Big)\Big)\Big].
\end{align}
% %%%%
% Note $ D_{k}\left(x^{n}\right) $ can be expressed by 
% \begin{align}
    %  D_{k}\left(x^{n}\right)
    %  =&\min _{g_k} E_{P_{x^{n} }^{k}\otimes \mathcal{C}} [\Lambda\left(X_{0}, g_k\left(Y_{-k}^{k}\right)\right)],\label{eq:dk_def1}\\
    %  =&E_{P_{x^{n} }^{k}\otimes \mathcal{C}}\left[U\left(\mathbf{P}\left(X_{0} | Y_{-k}^{k}\right)\right)\right].\label{eq:dk_def2}
%\end{align}
Now, we introduce a regularity assumption to carry out analysis for the performance bound. 

\begin{assumption}\label{asm2}
    Consider the network parameter $\mathbf{w}^*$ learned by minimizing $\mathcal{L}_{\texttt{Gen-CUDE}}$. Then, we assume there exists a sufficiently small $\epsilon'>0$ such that 
        % for all $\epsilon'>0$,  
    % Then, for the network parameter $\mathbf{w}^*$ learned by minimizing $\mathcal{L}_{\texttt{Gen-CUDE}}$,
    % we assume 
    % \begin{multline}
    $$
    \Big\|\mathbf{P}(Z_{0} | \mathbf{y}_{-0}^{(k)})-\mathbf{p}^k(\mathbf{w}^*,\mathbf{y}_{-0}^{(k)})\Big\|_{1}
    \leq\epsilon'
    % \leq \frac{\epsilon^{*} }{\sum_{i=1}^{M}\|\pi_{i}^{-1}\|_{2} },
    $$
    holds for all contexts $\mathbf{y}_{-0}^{(k)}\in\mathbb{R}^{2k}$.
    % for all context $\mathbf{y}_{-0}^{(k)}$ that appeared in $y^n$ in which $\pi_i^{-1}$ stands for the $i$-th column of $\mathbf{\Pi}^{-1}$.
\label{assumption_3}
\end{assumption}
Assumption \ref{asm2} is based on the universal approximation theorem \citep{Cybenko89,HorStiWhi89}, which ensures that there always exists a neural network that can approximate any function with arbitrary accuracy. Thus, we assume that the neural network learned by minimizing $\mathcal{L}_{\texttt{Gen-CUDE}}$ results in an accurate enough approximation of the true 
% a reasonable assumption that $\mathbf{p}^k(\mathbf{w}^*,\mathbf{y}_{-0}^{(k)})$ 
% the neural network obtained by minimizing $\mathcal{L}_{\texttt{Gen-CUDE}}$ 
% can approximate the true 
probability vector $\mathbf{P}(Z_{0} | \mathbf{y}_{-0}^{(k)})$. 
% very accurately.
% $ \mathbf{p}^k(\mathbf{w}^*,\mathbf{y}_{-0}^{(k)})$ that can approximates the probability function $\mathbf{P}(Z_{0} | \mathbf{y}_{-0}^{(k)})$ with arbitrarily small accuracy. 

Now, by letting 
\begin{align}
    \hat{\mathbf{P}}(X_0 | y_{-k}^{k}) \triangleq  \frac{p(\mathbf{y}_{-0}^{(k)}) }{p(y_{-k}^{k})} [\mathbf{p}^k(\mathbf{w}^*,\mathbf{y}_{-0}^{(k)}) \Pib^{-1}]\odot \mathbf{f}_{X_0}(y_0),\nonumber
\end{align}
% in which $\mathbf{p}^k(\mathbf{w}^*,\cdot)$ is defined in Section 4.2, 
we can then show from Assumption 2 that
\begin{equation}
    %  \left\|\mathbf{P}\left(X_{0} | y_{-k}^{k}\right)-\hat{\mathbf{P}}\left(X_{0} | y_{-k}^{k}\right)\right\|_{1} \leq \epsilon^{*}
    % \text{ for all $y_{-k}^{k}$.}
    \mathbb{E} \|\mathbf{P}(X_{0}|Y_{-k}^{k})-\hat{\mathbf{P}}(X_{0}|Y_{-k}^{k})\|_{1}\leq \epsilon^{*},
    \label{ineq:prop_assumption}
\end{equation}
for $\epsilon^*=\epsilon'\sum_{a=0}^{M-1}\|\pi_a^{-1}\|_2$,
in which $\mathbb{E}(\cdot)$ is the expectation with respect to  $Y_{-k}^k$, and $\pi_a^{-1}$ stands for the $a$-th column of $\mathbf{\Pi}^{-1}$.
The proof of (\ref{ineq:prop_assumption}) is given Lemma 2 in the Supplementary Material, and it plays an important role in proving the main theorem.

Before stating the main theorem, we first introduce $\mathcal{R}_{\delta}$, which is a quantizer that rounds each component of a probability vector to the nearest integer multiple of $\delta>0$ in $[0,1]$. 
% Then, we let 
% We also introduced $Q_{\delta}$ notation to restrict the function class which can be obtained from adjusting neural network output.
% Let $Q_{\delta}$ denotes quantization of each component to the nearest integer multiple of $\delta$ which is $\leq 1$ and $> 0$.
% \begin{lemma}
%     For $M>0$ and $\delta>0$, suppose that $\hat{\mathbf{P}}^{\delta}$ = $Q_{\delta}(\hat{\mathbf{P}})$. Then, 
%     \begin{equation}
%     \left\|\hat{\mathbf{P}}^{\delta}\left(X_{0} | y_{-k}^{k}\right)-\hat{\mathbf{P}}\left(X_{0} | y_{-k}^{k}\right)\right\|_{1} 
%     \leq \frac{M \cdot \delta}{2}.
%     \label{ineq:lemma_2}
%     \end{equation}
%     \label{lemma_2}
% \end{lemma}
% Let $ L_{\hat{X}_{\text{NN}}^{\delta}}\left(x^{n}, Y^{n}\right)$ denotes the average loss of denoiser
Then, consider a denoiser $\hat{X}^{n,\delta}_{NN}(Y^n)$ of which $i$-th component ($k\leq i\leq n-k$) is defined as 
$\hat{X}_{i,\text{NN}}^{\delta}(Y^n)=\mathcal{B}(\hat{\mathbf{P}}^{\delta}(X_i|Y_{i-k}^{i+k}))$,
where $\hat{\mathbf{P}}^{\delta}(X_i|Y_{i-k}^{i+k})=\mathcal{R}_{\delta}(\hat{\mathbf{P}}(X_i|Y_{i-k}^{i+k}))$. Note when $\delta$ is small enough, the performance of $\hat{X}^{n,\delta}_{\text{NN}}(Y^n)$ would be close to that of Gen-CUDE. Now, we have the following theorem.

% Note when $\delta$ is small enough the performance of $\{\hat{X}_{i,\text{NN}}^{\delta}(Y^n)\}_{i=k+1}^{n-k}$ will be ...

\begin{theorem}
\label{thm_1}
Consider $\epsilon^*$  in (\ref{ineq:prop_assumption}). Then, 
for all $k,n\geq 1$, $\delta>0$, and $\epsilon > \Lambda_{\max}\cdot ( 3\epsilon^* + \frac{M \cdot \delta}{2})$, and for all $x^n$,
\begin{align*}
    \Pr\Big(&| L_{\hat{X}_{\text{NN}}^{n,\delta}}\left(x^{n}, Y^{n}\right)
- D_{x^n}^k |> \epsilon  \Big)\\
\leq& C_1(k,\delta,M)
% 2(2k+1)\cdot\Big[\frac{1}{\delta}+1\Big]^{M} \\
\exp \Big(-\frac{2\left(n-2k\right)}{\left(2k+1\right)} C_2(\epsilon, \epsilon^*, \Lambda_{\max}, M, \delta) \Big),\nonumber
\end{align*}
in which $C_1(k,\delta,M)\triangleq2(2k+1)[\frac{1}{\delta}+1]^{M}$ and 
% $ \mathcal{D} $ is defined as 
$C_2(\epsilon, \epsilon^*, \Lambda_{\max}, M, \delta)\triangleq ( \epsilon-\Lambda_{\max}\cdot ( 3\epsilon^* + \frac{M \cdot \delta}{2}) )^2 \cdot \frac{1}{\Lambda_{\max }^{2}}$.

\end{theorem}
\emph{Proof:} The full proof of the theorem as well as necessary lemmas are given in the Supplementary Material. 

Theorem \ref{thm_1} states that for any $x^n$, with high probability, Gen-CUDE can essentially achieve the performance of the best sliding-window denoiser with the same order $k$.
Note that our bound has the constant term $[\frac{1}{\delta}+1]^{M}$, whereas the paralleling result in \citep{dembo2005universal} has $[\frac{1}{\delta}+1]^{M^{2k+1}}$. Removing such doubly exponential dependency on $k$ in our result is mainly due to our directly modeling the marginal posterior distribution via neural network, as opposed to the modeling of the joint posterior of the $(2k+1)$-tuple in the previous work. 
% whereas \citep{dembo2005universal} approximated 
% the cardinality of function class utilized to Gen-CUDE is $[\frac{1}{\delta}+1]^{M}$. Compared to the approach of \citep{dembo2005universal} which deals with the function class of cardinality $[\frac{1}{\delta}+1]^{M^{2k+1}}$, we removed the dependency of $k$ in the upper index part of cardinality. 
This improvement carries over to the better empirical performance of the algorithm given in the next section.
\begin{figure*}[t]
\centering
\subfigure[][Denoising Performace]{
\includegraphics[width=2.0\columnwidth]{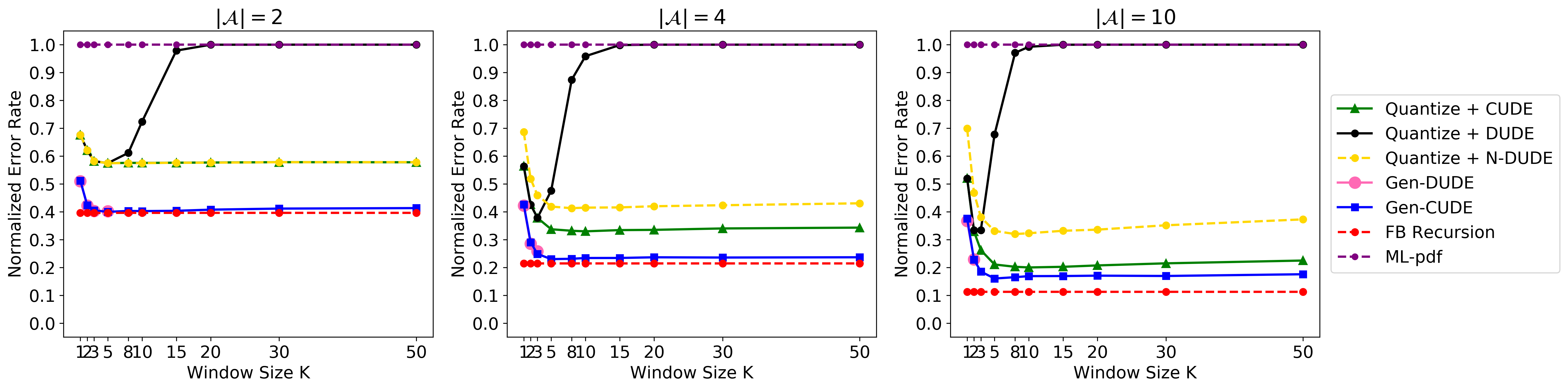}
\label{fig1a}
}
\subfigure[][Running Time]{
\includegraphics[width=2.0\columnwidth]{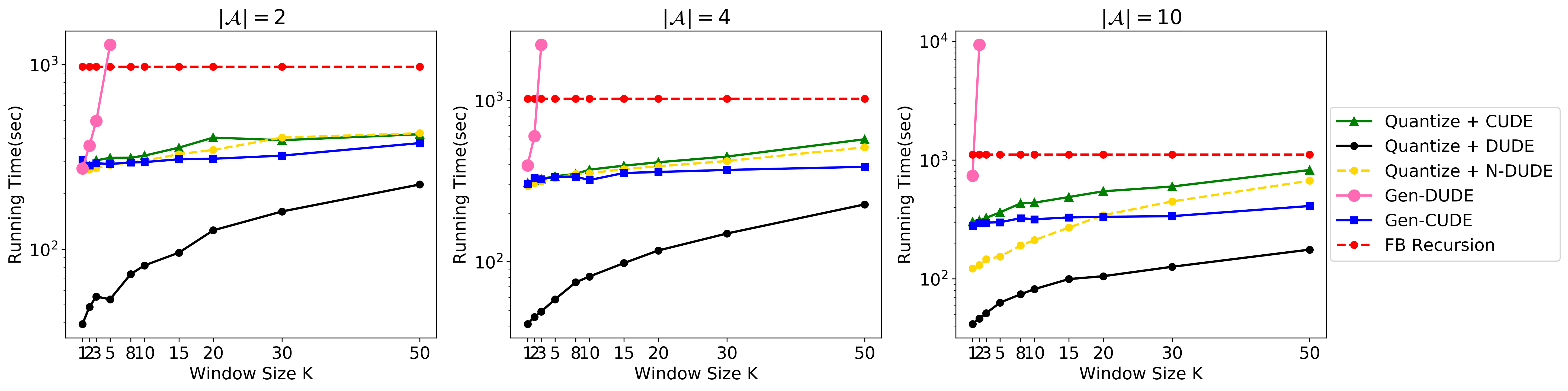}
\label{fig1b}
}
\caption{\label{fig1)}
Synthetic source data case. (a) Denoising results and (b) running time (including training time). The left, center, and right plots correspond to the case of $|\mcA|=2$, $4$, and $10$, respectively.
% For algorithms that utilize neural networks, the running time also includes the training time.
}\vspace{-.1in}
\end{figure*}
%width=1.00\columnwidth
%scale=0.35
%%%%%%%%%%%%%%%%%%%%%%%%%%%%%%%%%%%%%%%%%%%%%%%%%%%%%%%%%%%%%%%%%%%%%%%%%%%%%%

\section{Experimental Results}
\subsection{Setting and baselines}
% \subsection{Comparison for Different Algorithm with Simple Quantization}
% \subsection{Comparison for Different Number of Decision Boundary}

%Setting 및 디테일
%우리는 우리 알고리즘의 성능과 효율성을 보여주기 위해 간단한 synthetic source data와 real DNA source 데이터를 활용하였다. 
We have experimented with both synthetic and real DNA source data and verified the effectiveness of our proposed Gen-CUDE algorithm. The noisy channel $\mathcal{C}=\{f_a\}_{a\in\mcA}$ was assumed to be known, and the noisy observation $Y^n$ was generated by corrupting the source sequence $x^n$. We used the Hamming loss as our $\Lb$ to measure the denoising performance.

% For the quantizer 

% , and the noisy channel $\mathcal{C}=\{f_a\}_{a\in\mcA}$ is assumed to be known to the denoiser. 

We have compared the performance of \texttt{Gen-CUDE} with several baselines. The simplest baseline is \texttt{ML-pdf}, which carries out the symbol-by-symbol maximum likelihood estimate, i.e., $\hat{X}_i(Y^n)=\arg\max_a f_a(Y_i)$. 
% Hence, it completely ignores the redundancy that exists in the source sequence $x^n$.
The other baselines are schemes that apply discrete denoising algorithms on the quantized $Z^n$ using the induced DMC $\Pib$. That is, these schemes simply throw away the continuous-valued observation $Y^n$ and the density values. We denoted such schemes as \texttt{Quantized$+$DUDE}, \texttt{Quantized$+$N-DUDE}, and \texttt{Quantized$+$CUDE}. We also employed \texttt{Gen-DUDE} as a baseline for FIGO channel.
% and showed the superiority of our \texttt{Gen-CUDE}. 
For neural network training, we used a fully-connected network with ReLU \citep{nair2010rectified} activations
% and He-normal initialization \citep{he2015delving}
. For more details on the implementation, the code is available online\footnote{\texttt{https://github.com/pte1236/Gen-CUDE}}.

\subsection{Synthetic source with Gaussian noise}
% \subsubsection{Binary}
% \subsubsection{Tetra}
% \subsubsection{Deca}

%%%%%%%%%%%%%%%%%%%%%%%%%%%%%%%%%%%%%%%%%%%%%%%%%%%%%%
% synthetic source data 파트
% synthetic source data로는 symmetric markov chain source data을 활용하였다

For the synthetic source data case, we generated the clean sequence $x^n$ from a symmetric Markov chain. We varied the alphabet size $|\mcA|=2,4,10$, and the source symbol was encoded to have odd integer values $\mathcal{O}=\{\pm (2\ell-1):1\leq \ell \leq |\mcA|/2\}$. The transition probability of the Markov source was set to $0.9$ for staying on the same state and $0.1/|\mcA|$ for transitioning to the other state. The sequence length was $n=3\times 10^6$, and the noisy channel was set to be the standard additive white Gaussian, $\mathcal{N}(0,1)$. The neural network had 6 fully-connected layers and 200 nodes in each layer.
For the quantizer $Q(\cdot)$ in all of our experiments, 
% in this paper except Appendix D, 
we simply rounded to the nearest integer among $\mathcal{O}$. Note that $Q(\cdot)$ can be freely selected for \texttt{Gen-CUDE} as long as the induced DMC, $\mathbf{\Pi}$, is invertible, and we show the little effect of the choice of $Q(\cdot)$ on the denoising performance 
% in additional experiments given 
in the Supplementary Material. 
% To show the little effect of the quantizer on the final denoising performance, we designed two additional experiments for the $|\mathcal{A}|=4$ case in Appendix D. 

The denoising performance as well as the running time of each scheme is given in Figure \ref{fig1)}, and the performance in Figure \ref{fig1a} was normalized with the performance of the simple quantizer, $\hat{X}_i=Z_i=Q(Y_i)$. Note for the symmetric Gaussian noise, \texttt{ML-pdf} becomes equivalent to applying $Q(\cdot)$, but they can become different for general noise densities. Moreover, we compared the performance with \texttt{FB-Recursion} \citep[Section V.]{EphMer02}, which is the optimal
scheme for the given setting since the noisy sequence becomes a hidden Markov process (HMP). 
% That is, \texttt{FB-Recursion} exactly knows the source Markov model and its transition probability as well as the channel noise density, hence, it gives the lower bound to the normalized error rate. 

\begin{figure*}[h]
\centering\includegraphics[width=2.0\columnwidth]{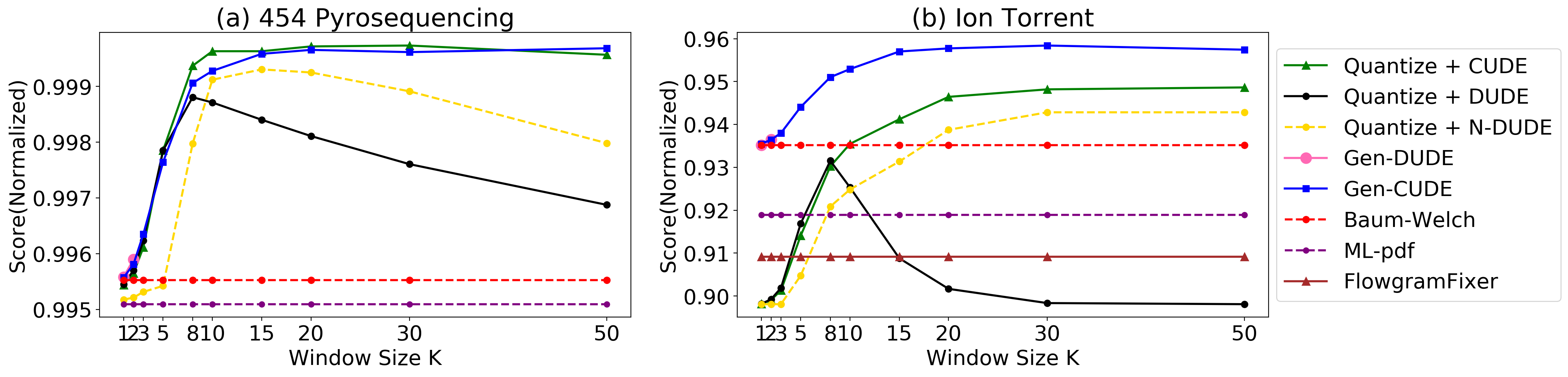}
\caption{\label{fig2}
Normalized similarity scores against the clean reference DNA sequence for both sequencing platforms. 
% The scores are computed for each query sequence, and 
% Bit error rate (left) and score measured by alignment tool (right) for 454 Pyrosequencing source data (a) and Ion Torrent source data (b). Bit error rate results for each algorithms are normalized with the bit error rate of simple quantization method. In the case of score graph, we calculated similarity score between DNA-form sequences of real source and denoised sequence via alignment tool. We used pairwise2 module of Biopython as the alignment tool. By accumulating the score and maximum length between two objects, we normalized the sum of query-wise scores with sum of maximum lengths.
}
\vspace{-.2in}
\end{figure*}

From the figures, we can make several observations. Firstly, we note that the neural network-based schemes, i.e., \texttt{Quantize$+$N-DUDE}, \texttt{Quantize$+$CUDE}, and our \texttt{Gen-CUDE},  are very robust with respect to the window size $k$. The effect of the window size $k$ for \texttt{Gen-CUDE} not being huge compared to \texttt{Gen-DUDE} can be predicted from the bound in Theorem 1. In contrast, \texttt{Quantize$+$DUDE} becomes quite sensitive to $k$ as has been identified in \citep{moon2016neural}. Secondly, we observe our \texttt{Gen-CUDE} always achieves the best denoising performance among the baselines and gets close to the optimal \texttt{FB-Recursion}. Note \texttt{Gen-CUDE} knows nothing about the source sequence $x^n$, whereas \texttt{FB-Recursion} exactly knows the source Markov model. Moreover, while \texttt{Gen-DUDE} performs almost as well as \texttt{Gen-CUDE} for $|\mcA|=2$ with appropriate $k$, its performance significantly deteriorates when the alphabet size grows. We see that the gap between \texttt{Gen-CUDE} and  \texttt{FB recursion} widens (although not much) as the alphabet size $M$ increases, which can also be predicted from the bound in Theorem 1. Thirdly, \texttt{Gen-DUDE} suffers from the prohibitive computational complexity as $k$ grows, as shown in Figure \ref{fig1b}, while the running time of our \texttt{Gen-CUDE} is orders of magnitude faster than that of \texttt{Gen-DUDE} and more or less constant with respect to $k$. From this reason, \texttt{Gen-DUDE} can be run only for small $k$ values. Fourthly, we note \texttt{Quantize$+$CUDE} also performs reasonably well, and it outperforms all discrete denoising baselines as also shown in \citep{ryu2018conditional}. However, since it discards the additional soft information in the continuous-valued observation and density values, \texttt{Gen-CUDE}, which is tailored for the FIGO channel, outperforms \texttt{Quantize$+$CUDE} with a significant gap.

\vspace{-.1in}

\subsection{DNA source with homopolymer errors}

Now, we verify the performance of \texttt{Gen-CUDE} on real DNA sequencing data. 
We focus on the homopolymer errors which is the dominant
error type in “sequencing by synthesis” methods, e.g., Roche 454 pyrosequencing \citep{quince2011removing} or Ion Torrent
Personal Genome Machine (PGM) \citep{bragg2013shining}.
In those methods, each nucleotide in turn is iteratively washed over with a pre-determined short sequence of bases known as the “wash cycle”, and the continuous-valued flowgrams are observed. 
Recently, \citep{LeeMooYooWei16} describes how we can interpret the base-calling procedure of such sequencers exactly as our FIGO channel setting by mapping the DNA sequence into sequence of integers (of homopolymer length), which becomes the input to the noisy channel $\{f_a\}_{a\in\mcA}$, the flowgram densities for each homopolymer length. This interpretation is possible since the order of the nucleotides in the wash cycle is fixed. Denoising in such setting can correct insertion and deletion errors, the dominant and notoriously hard types of errors in such sequencers.

We used \texttt{Artificial.dat}, a public dataset used in \citep{quince2011removing} for 454 pyrosequencing, and \texttt{IonCode0202.CE2R.raw.dat}, a data obtained from an internal source, for Ion Torrent after preprocessing both datasets. We \emph{simulated} the channel of each platform to obtain the noisy sequence that is corrupted with the homopolymer errors. The used noisy channel density was provided for 454, but not for Ion Torrent, hence, we estimated the density for Ion Torrent with a small holdout set with clean source using Gaussian kernel density estimation with bandwidth 0.6. The used densities for 454 and Ion Torrent are shown in the Appendix D of the Supplementary Material. 
% was provided for 454, but not for Ion Torrent, hence, we estimated the density for Ion Torrent with small noisy-clean pairs using Gaussian kernel density estimation. The used 
% The noisy channel density 
% 454,iontorrent의 전체길이는 각각 11403764 , 6819480 이다
% The total sequence length $n$ of 454 and Ion Torrent is 11,403,764 and 6,819,480, respectively.
% % 여기서, 앞에서부터 쿼리인덱스를 약 60%정도만 가져와서 디노이징을 수행했으며, 쿼리를 중간에서 자르는 일이 없도록 하였다.
% From the beginning, only about 60\% of the query index was taken to be the set for denoising. So, the query was not cut in the middle.
% 결과적으로 디노이징에 사용되는 데이터의 길이는 각각 6845624,4101244이다
The total sequence length $n$ for 454 and Ion Torrent data was 6,845,624 and 4,101,244, respectively.
%ion flowcycle은 32개다
Moreover, the wash cycle of 454 and Ion Torrent was \texttt{TACG} and \texttt{TACGTACGTCTGAGCATCGATCGATGTACAGC}, respectively.
% 454에서는 채널정보가 주어졌지만, ion의 경우 454와 달리, 채널정보가 주어져 있지 않았다. 
We set the maximum homopolymer length to 9, hence, the source symbol can take values among $\{0,\ldots,9\}$. 
% The noisy channel density was provided for 454, but not for Ion Torrent. Thus, for Ion Torrent, we estimated the channel in separate small dataset with noisy-clean sequence pairs using Gaussian kernel density estimation with bandwidth $0.6$. 
% Both channel densities are given in the Supplementary Material. 
The neural network for \texttt{Gen-CUDE} had 7 layers with 500 nodes in each layer. After denoising, the error correction performance was compared with the similarity score between the clean reference sequence and the denoised sequence (after converting back from the sequence of homopolymer lengths to the DNA sequence). 
The score was computed with the Pairwise2 module of Biopython, a common alignment tool to compute the similarity between DNA sequences \citep{chang2010biopython}.

Figure \ref{fig2} shows the error correction performance for both 454 and Ion Torrent platform data. The score is normalized so that 1 corresponds to the perfect recovery. We also included \texttt{Baum-Welch} \citep{BaumPetrie70} that treats the source as a Markov source and estimates the transition probability before applying the \texttt{FB-recursion}. We can make the following observations. Firstly, we note that \texttt{Baum-Welch} is no longer optimum since the source DNA sequence is far from being a Markov. Secondly, \texttt{Quantize$+$DUDE}, which was used to correct the homopolymer errors in \citep{LeeMooYooWei16}, turns out to be suboptimal, as also observed in Figure \ref{fig1)}. Thirdly,  \texttt{Gen-CUDE} again achieves the best error correction performance for both 454 and Ion Torrent data. Note for 454, in which the original error rate is quite small, the performance of \texttt{Quantize$+$CUDE} and \texttt{Gen-CUDE} becomes almost indistinguishable, but in Ion Torrent, of which noise density has a higher variance and non-zero means, the performance gap between the two widens. Fourthly, in line with Figure \ref{fig1)}, we were not able to run \texttt{Gen-DUDE} for more than $k=2$.
% due to the memory and computation issues. 
Finally, for Ion Torrent, we also run \texttt{FlowgramFixer}~\citep{golan2013using}, the state-of-the-art homopolymer error correction tool for Ion Torrent, but it showed the worst performance.

\vspace{-.1in}

\section{Conclusion}
% % 우리 알고리즘은 figo setting에서 k가 커지면 대체로 성능이 증가하고, k가 커질떄 시간이 오래걸리는 general dude의 문제를 해결하였다.

We devised a novel unsupervised neural network-based \texttt{Gen-CUDE} algorithm, which carries out universal denoising for FIGO channel. Our algorithm was shown to significantly outperform previously developed algorithm for the same setting, \texttt{Gen-DUDE}, both in denoising performance and computation complexity. We also give a rigorous theoretical analyses on the scheme and obtain a tighter upper bound on the average error compared to \texttt{Gen-DUDE}. Our experimental results show promising results, and as a future work, we plan to apply our method to \emph{real} noisy data denoising and make more algorithmic improvements, e.g., using adaptive quantizers instead of simple rounding. 

\section*{Acknowledgement}
This work is supported in part by Institute of Information \& communications Technology Planning  Evaluation (IITP) grant funded by the Korea government (MSIT) [No.2016-0-00563, Research on adaptive machine learning technology development for intelligent autonomous digital companion], [No.2019-0-00421, AI Graduate School Support Program (Sungkyunkwan University)], [No.2019-0-01396, Development of framework for analyzing, detecting, mitigating of bias in AI model and training data], and [IITP-2019-2018-0-01798, ITRC Support Program]. The authors also thank Seonwoo Min, Byunghan Lee and Sungroh Yoon for their helpful discussions on DNA sequence denoising, and thank Jae-Ho Shin for providing the raw Ion Torrent dataset.

\bibliography{bibfile}

\clearpage

\setcounter{lemma}{1}
\setcounter{figure}{2}
\setcounter{section}{0}

\newcounter{counting_number}
\addtocounter{counting_number}{15}

\onecolumn
\aistatstitle{Supplementary Material for \\Unsupervised Neural Universal Denoiser for Finite-Input General-Output Noisy Channel}
\aistatsauthor{ Tae-Eon Park and Taesup Moon}
% \And Byunghan Lee \And  Seonwoo Min, Sungroh Yoon}
\aistatsaddress{Department of Electrical and Computer Engineering\\ Sungkyunkwan University (SKKU), Suwon, Korea 16419 \\ \texttt{\{pte1236, tsmoon\}@skku.edu}
%  \And  Seoul National University of \\ Science and Technology \\ Seoul, Korea 01811 \\ \texttt{bhlee@seoultech.ac.kr}
%  \And Seoul National University \\ Seoul, Korea 08826 \\ \texttt{\{mswzeus, sryoon\}@snu.ac.kr}
 }

%%%%%%%%%%%%%%%%%%%%%%%-----------------------------supplementary

% \appendix
% \onecolumn

% \aistatstitle{Supplementary Material \\ for Unsupervised Neural Universal Denoiser for \\Finite-Input General-Output Noisy Channel}
% \aistatsauthor{Anonymous Authors}
% \aistatsaddress{Unknown Institution}

% {\centering{\LARGE\bfseries Unsupervised Neural Universal Denoiser for Finite-Input General-Output Noisy Channel}

% \vspace{1em}
% \centering{{\LARGE\bfseries Supplementary Material}}

% }

% {\centering
% \vspace{1em}
% \centering{{\LARGE\bfseries Supplementary Material for \\
% Unsupervised Neural Universal Denoiser for Finite-Input General-Output Noisy Channel}}
% }

\gdef\thesection{Appendix \Alph{section}}
\section{Proof of Theorem 1}\label{appendix:proof}
The following lemma formalizes to prove Eq. (16) in the paper. First, let  
$    \epsilon^*=\epsilon'\sum_{a=0}^{M-1}\|\pi_a^{-1}\|_2
$
as shown in the paper and define
\begin{align}
    \hat{\mathbf{P}}(X_0 | y_{-k}^{k}) \triangleq  \frac{p(\mathbf{y}_{-0}^{(k)}) }{p(y_{-k}^{k})} \cdot [\mathbf{p}^k(\mathbf{w}^*,\mathbf{y}_{-0}^{(k)}) \cdot \Pib^{-1}]\odot \mathbf{f}_{X_0}(y_0).\label{eq:p_hat}
\end{align}

\begin{lemma}
Suppose network parameter $\mathbf{w}^*$ learned by minimizing $\mathcal{L}_{\texttt{Gen-CUDE}}$ satisfies Assumption 2. Then,
% for $\epsilon^*=\epsilon'\sum_{a=0}^{M-1}\|\pi_a^{-1}\|_2$, the following inequality holds:
\begin{equation*}
     \mathbb{E} \|\mathbf{P}(X_{0}|Y_{-k}^{k})-\hat{\mathbf{P}}(X_{0}|Y_{-k}^{k})\|_{1}\leq \epsilon^{*},
\end{equation*}
in which the expectation is with respect to $Y_{-k}^k$.
% where $\hat{\mathbf{P}}(X_0 | y_{-k}^{k}) \triangleq  \frac{p(\mathbf{y}_{-0}^{(k)}) }{p(y_{-k}^{k})} \cdot [\mathbf{p}^k(\mathbf{w}^*,\mathbf{y}_{-0}^{(k)}) \cdot \Pib^{-1}]\odot \mathbf{f}_{X_0}(y_0)$.
\label{lemma_assumption2_to_15}
\end{lemma}

\begin{proof}
We have the following chain of equations.
    \begin{align}
    &  \mathbb{E} \|\mathbf{P}(X_{0}|Y_{-k}^{k})-\hat{\mathbf{P}}(X_{0}|Y_{-k}^{k})\|_{1} = \int_{\mathbb{R}^{2k+1}} p(y_{-k}^{k})\cdot \| \mathbf{P}(X_{0}|y_{-k}^{k})-\hat{\mathbf{P}}(X_{0}|y_{-k}^{k}) \|_{1} \, d y_{-k}^{k} \nonumber\\
     = \, & \int_{\mathbb{R}^{2k+1}}  p(y_{-k}^{k})\cdot  \Big\| \frac{p(\mathbf{y}_{-0}^{(k)}) }{p(y_{-k}^{k})} \cdot [ \Big( \mathbf{P}(Z_{0}|\mathbf{y}_{-0}^{(k)})-\mathbf{p}^k(\mathbf{w}^*,\mathbf{y}_{-0}^{(k)}) \Big) \cdot \Pib^{-1} ] \odot \mathbf{f}_{X_0}(y_0) \Big\|_{1} \, d y_{-k}^{k} \label{eq:lem1_1} \\
     = \, & \int_{\mathbb{R}^{2k+1}}  p(y_{-k}^{k})\cdot \Big[ \sum_{a=0}^{M-1} \Big|  \Big( \mathbf{P}(Z_{0}|\mathbf{y}_{-0}^{(k)})-\mathbf{p}^k(\mathbf{w}^*,\mathbf{y}_{-0}^{(k)}) \Big) \cdot \pi_{a}^{-1} \cdot {f}_{a}(y_0) \Big| \Big] \cdot \frac{p(\mathbf{y}_{-0}^{(k)}) }{p(y_{-k}^{k})} \, d y_{-k}^{k}  \nonumber\\
     = \, & \sum_{a=0}^{M-1} \int_{\mathbb{R}^{2k+1}}  \Big|  \Big( \mathbf{P}(Z_{0}|\mathbf{y}_{-0}^{(k)})-\mathbf{p}^k(\mathbf{w}^*,\mathbf{y}_{-0}^{(k)}) \Big) \cdot \pi_{a}^{-1} \Big| \cdot {f}_{a}(y_0) \cdot  p(\mathbf{y}_{-0}^{(k)})  \, d y_{-k}^{k}  \nonumber\\
     \leq \, & \sum_{a=0}^{M-1} \int_{\mathbb{R}^{2k+1}}  \Big\|   \mathbf{P }(Z_{0}|\mathbf{y}_{-0}^{(k)})-\mathbf{p}^k(\mathbf{w}^*,\mathbf{y}_{-0}^{(k)}) \Big\|_{2} \cdot \| \pi_{a}^{-1} \|_{2}  \cdot {f}_{a}(y_0) \cdot  p(\mathbf{y}_{-0}^{(k)})  \, d y_{-k}^{k}  \label{eq:lem1_2}\\
     \leq \, &  \sum_{a=0}^{M-1} \int_{\mathbb{R}^{2k+1}}  \Big\|   \mathbf{P }(Z_{0}|\mathbf{y}_{-0}^{(k)})-\mathbf{p}^k(\mathbf{w}^*,\mathbf{y}_{-0}^{(k)}) \Big\|_{1} \cdot \| \pi_{a}^{-1} \|_{2}  \cdot {f}_{a}(y_0) \cdot  p(\mathbf{y}_{-0}^{(k)})  \, d y_{-k}^{k}  \label{eq:lem1_3}\\
     \leq \, &  \epsilon' \cdot  \sum_{a=0}^{M-1} \Big[ \| \pi_{a}^{-1} \|_{2} \cdot \int_{\mathbb{R}^{2k+1}} {f}_{a}(y_0) \cdot  p(\mathbf{y}_{-0}^{(k)})   \, d y_{-k}^{k} \Big]  \label{eq:lem1_4}\\
     = \, & \epsilon' \cdot  \sum_{a=0}^{M-1} \Big[ \| \pi_{a}^{-1} \|_{2} \cdot \int_{\mathbb{R}} {f}_{a}(y_0) \int_{\mathbb{R}^{2k}}  p(\mathbf{y}_{-0}^{(k)}) \, d \mathbf{y}_{-0}^{(k)} \, d y_{0} \Big]   =  \epsilon' \cdot \sum_{a=0}^{M-1} \| \pi_{a}^{-1} \|_{2}   = \epsilon^{*},
     \end{align}
     in which (\ref{eq:lem1_1}) follows from Lemma 1 and (\ref{eq:p_hat}), (\ref{eq:lem1_2}) follows from the Cauchy-Schwarz inequality, and (\ref{eq:lem1_3}) follows from the fact that $L_2$-norm is smaller than the $L_1$-norm, and (\ref{eq:lem1_4}) follows from Assumption 2.\ \ \qed
\end{proof}

\clearpage

\begin{lemma}
    Let $\mathcal{R}_{\delta}(\cdot)$ denote the quantizer that rounds each component of the argument probability vector to the nearest integer multiple of $\delta$ in $(0,1]$. 
    For $M>0$ and $\delta>0$, denote $\hat{\mathbf{P}}^{\delta}$ = $\mathcal{R}_{\delta}(\hat{\mathbf{P}})$. Then, 
    \begin{equation*}
    \left\|\hat{\mathbf{P}}^{\delta}\left(X_{0} | y_{-k}^{k}\right)-\hat{\mathbf{P}}\left(X_{0} | y_{-k}^{k}\right)\right\|_{1} 
    \leq \frac{M \cdot \delta}{2}.
    \end{equation*}
    \label{lemma_quan}
\end{lemma}
\begin{proof}
    By the definition of $\hat{\mathbf{P}}^{\delta}$, it is clear that
    \begin{align*}
    & \left\|\hat{\mathbf{P}}^{\delta}\left(X_{0} | y_{-k}^{k}\right)-\hat{\textbf{P}}\left(X_{0} | y_{-k}^{k}\right)\right\|_{\infty} 
    \leq \frac{\delta}{2}. \\
    \text{Therefore, } & \left\|\hat{\mathbf{P}}^{\delta}\left(X_{0} | y_{-k}^{k}\right)-\hat{\textbf{P}}\left(X_{0} | y_{-k}^{k}\right)\right\|_{1} 
    = \sum_{a=0}^{M-1}|\hat{p}^{\delta}\left(a | y_{-k}^{k}\right)-\hat{p}(a | y_{-k}^{k})|
    \leq \sum_{a=0}^{M-1} \frac{\delta}{2} \leq \frac{M\delta}{2}. \ \ \qed
    \end{align*}
\end{proof}

From Lemma \ref{lemma_quan}, we can expect that for the sufficiently small $\delta$, performance of the denoisers using $\hat{\mathbf{P}}^{\delta}$ and $\hat{\mathbf{P}}$ respectively, for computing the Bayes response will be close to each other.

% can be close each other with respect to the $L^1$ metric.

% $\left\|\textbf{P}\left(Z_{0} | y_{-k}^{-1},y_{1}^{k}\right)-\hat{\textbf{P}}\left(Z_{0} | y_{-k}^{-1},y_{1}^{k}\right)\right\|_{1}$ \\ 
%     $\leq$ $\frac{\epsilon}{\left(\sum_{i=1}^{M}\left\|\pi_{a}^{-1}\right\|_{2}\right) \cdot \left\|\textbf{f}_{U_{0}}\left(y_{0}\right)\right\|_{\infty}}$

% \begin{lemma}
% lemma for sufficiently small $\delta$ takes same result, same loss. (example : hamming loss ?)
% \label{lemma_2}
% \end{lemma}

% \begin{lemma}
% For $\textbf{P}, \hat{\mathbf{P}} \in \Delta^M $ and Bayes envelope U,
% $$|D_{k}(x^{n}) - \mathbb{E}_{P_{x^n}^k \otimes \mathcal{C}} [U\left(\hat{\mathbf{P}}({X_{0} | Y_{-k}^{k})}\right)]| \leq \Lambda_{\max} \cdot \left(\max_{y_{-k}^{k}} \|\textbf{P}\left(X_{0}| y_{-k}^{k}\right)-\hat{\mathbf{P}}\left(X_{0} | y_{-k}^{k}\right) \|_{1}\right).$$
% \end{lemma}
\begin{lemma}
Consider $\hat{\mathbf{P}}(X_{0} | Y_{-k}^{k})$ and $\epsilon^*$ defined in Lemma 2 and the performance target $D_{x^n}^k$ defined in Eq. (15) in the paper. Then, we have
$$
\Big|D_{x^n}^k - \mathbb{E}_{P_{x^n}^k \otimes \mathcal{C}} \Big[U\Big(\hat{\mathbf{P}}({X_{0} | Y_{-k}^{k})}\Big)\Big]\Big|\leq \Lambda_{\max}\cdot \epsilon^*,
$$
\label{lemma_ex_est}
in which $P_{x^n}^k \otimes \mathcal{C}$ stands for the joint distribution on $(X_0,Y_{-k}^k)$ defined by the empirical distribution $P_{x^n}^k(u_{-k}^k) = \frac{1}{n-2k}\mathbf{r}[x^n,u_{-k}^k]$ with
% which shows the closeness of $D_{x^n}^k$ and the expected Bayes envelop for $\hat{\mathbf{P}}(X_{0} | Y_{-k}^{k})$.
$\mathbf{r}[x^{n}, u_{-k}^{k}]=|\{k+1 \leq i \leq n-k: x_{i-k}^{i+k}=u_{-k}^{k}\}|$  and the channel density $\mathcal{C}$.
\end{lemma}
\paragraph{Proof:}
First, we identify that
    \begin{align}
        & \left|D_{x^n}^k - \mathbb{E}_{P_{x^n}^k \otimes \mathcal{C}} \Big[U\Big(\hat{\mathbf{P}}({X_{0} | Y_{-k}^{k})}\Big)\Big]\right| \nonumber\\
        =& \left|\mathbb{E}_{P_{x^n}^k \otimes \mathcal{C}} \Big[U\Big(\textbf{P}({X_{0} | Y_{-k}^{k})}\Big) - \mathbb{E}_{P_{x^n}^k \otimes \mathcal{C}} [U\Big(\hat{\mathbf{P}}({X_{0} | Y_{-k}^{k})}\Big)\Big]\right| \label{eq:e_def}\\
      = & \int_{\mathbb{R}^{2k+1}} {\mathbb{E}}  \Big[\Lb \Big(X, \mathcal{B}(\hat{\mathbf{P}}(X_{0}|y_{-k}^{k}))\Big)
    -\Lb\Big(X,\mathcal{B}(\mathbf{P}(X_{0}|y_{-k}^{k}))\Big) \Big| y_{-k}^{k} \Big]
\cdot p\left(y_{-k}^{k}\right)  d y_{-k}^{k}, \label{eqn_e}
\end{align}
where the $\mathbb{E}(\cdot)$ in (\ref{eqn_e}) stands for the conditional expectation with respect to $\mathbf{P}(X_0|y_{-k}^k)$, which is the posterior distribution induced from $P_{x^n}^k \otimes \mathcal{C}$. Now, the following inequality holds for each $y_{-k}^{k}$:
\begin{align}
& {\mathbb{E}}\Big[\Lb \Big(X, \mathcal{B}(\hat{\mathbf{P}}(X_{0}|y_{-k}^{k}))\Big)
    -\Lb\Big(X, \mathcal{B}(\mathbf{P}(X_{0}|y_{-k}^{k}))\Big)\Big| y_{-k}^{k} \Big]\\
  = & \sum_{a=0}^{M-1}\mathbf{P}(X_0=a | y_{-k}^{k})
    \cdot \Big[\Lb\Big(a, \mathcal{B}(\hat{\mathbf{P}}(X_{0}|y_{-k}^{k}))\Big)
    -\Lb\Big(a, \mathcal{B}(\mathbf{P}(X_{0}|y_{-k}^{k}))\Big)\Big]\\
    \leq & \sum_{a=0}^{M-1} \Big( \mathbf{P}(X_0=a | y_{-k}^{k})-\hat{\mathbf{P}}(X_0=a | y_{-k}^{k}) \Big) 
    \cdot \Big[\Lb\Big(a, \mathcal{B}(\hat{\mathbf{P}}(X_{0}|y_{-k}^{k}))\Big)
    -\Lb\Big(a, \mathcal{B}(\mathbf{P}(X_{0}|y_{-k}^{k}))\Big)\Big] \label{eqn_f} \\
     \leq & \sum_{a=0}^{M-1} \Big|\Big( \mathbf{P}(X_0=a | y_{-k}^{k})-\hat{\mathbf{P}}(X_0=a | y_{-k}^{k}) \Big) \Big| \cdot \Lambda_{\max } = \Lambda_{\max } \cdot \| \textbf{P}(X_{0}|y_{-k}^{k})-\hat{\mathbf{P}}(X_{0}|y_{-k}^{k})\|_{1},
    \end{align}
in which (\ref{eqn_f}) follows from the definition of the Bayes response. Therefore, 
 \begin{align}
     (\ref{eqn_e}) \leq \, &  \Lambda_{\max}\cdot \int_{\mathbb{R}^{2k+1}}  \| \textbf{P}(X_{0}|y_{-k}^{k})-\hat{\mathbf{P}}(X_{0}|y_{-k}^{k})\|_{1}
\cdot p\left(y_{-k}^{k}\right)  d y_{-k}^{k}
=  \Lambda_{\max}\cdot {\mathbb{E}} \|\textbf{P}(X_{0}|Y_{-k}^{k})-\hat{\mathbf{P}}(X_{0}|Y_{-k}^{k})\|_{1} \label{eqn_r}\\
\leq& \Lambda_{\max}\cdot \epsilon^*  \label{eqn_s},
 \end{align}
in which (\ref{eqn_s}) follows 
% where $\mathbb{E}(\cdot)$ is expectation over $P_{x^n}^k \otimes \mathcal{C}$ and ${\mathbf{P}}(X_{0}|y_{-k}^{k})$ is a distribution induced from $P_{x^n}^k \otimes \mathcal{C}$. The last inequality follows 
from Lemma \ref{lemma_assumption2_to_15}. Note that difference between two expected loss of denoiser based on the Bayes response is bounded with the difference between two probability vectors.\ \ \qed

% \begin{lemma}
% For $\mathbf{P}, \hat{\mathbf{P}} \in \Delta^M $ and Bayes envelope U,
%     \begin{align*}
% |\mathbb{E}_{P_{x^n}^k \otimes \mathcal{C}}[&U\left(\hat{\mathbf{P}}\left(X_{0} | Y_{-k}^{k}\right)\right)] \\
% &- \mathbb{E}_{P_{x^n}^k \otimes \mathcal{C}}[U\left(\hat{\mathbf{P}}^{\delta}\left(X_{0} | Y_{-k}^{k}\right)\right)]|\\
% \leq 2\Lambda_{\max}\cdot \Big( \max_{y_{-k}^{k}}& \Big\{\|\mathbf{P}(X_{0}|y_{-k}^{k})-\hat{\mathbf{P}}(X_{0}|y_{-k}^{k})\|_{1}\Big\}\Big) \\
% &+ \frac{M\cdot \delta}{2} \leq 2\Lambda_{\max}\cdot \epsilon^*+ \frac{M\cdot \delta}{2}.
%     \end{align*}
% \label{lemma_5}
% \end{lemma}

\clearpage

\begin{lemma}
Consider $\hat{\mathbf{P}}(X_{0} | Y_{-k}^{k})$ and $\epsilon^*$ defined above
and define $\hat{\mathbf{P}}^{\delta}=\mathcal{R}_{\delta}(\hat{\mathbf{P}})$ as in Lemma 3. Then,

% For $\hat{\mathbf{P}} \in \Delta^M $, $\hat{\mathbf{P}}^{\delta}$ = $Q_{\delta}(\hat{\mathbf{P}})$, $\epsilon^* >0$, $\delta >0$ and Bayes envelope U,
    \begin{equation*}
\Big|\mathbb{E}_{P_{x^n}^k \otimes \mathcal{C}}\Big[U\Big(\hat{\mathbf{P}}(X_{0} | Y_{-k}^{k})\Big) - U\Big(\hat{\mathbf{P}}^{\delta}\Big(X_{0} | Y_{-k}^{k})\Big)\Big]\Big|
\leq 2\Lambda_{\max}\cdot \Big( \epsilon^*+ \frac{M\cdot \delta}{4}\Big).   
    \end{equation*}
\label{lemma_ex_quan}
\end{lemma}

\begin{proof}
We have the following chain of inequalities:
\begin{align}
& \Big|\mathbb{E}_{P_{x^n}^k \otimes \mathcal{C}} \Big[U\Big(\hat{\mathbf{P}}\left(X_{0} | Y_{-k}^{k}\right)\Big)\Big] 
- \mathbb{E}_{P_{x^n}^k \otimes \mathcal{C}}\Big[U\Big(\hat{\mathbf{P}}^{\delta}\left(X_{0} | Y_{-k}^{k}\right)\Big)\Big]\Big|\nonumber\\
\leq \, & \Big|\mathbb{E}_{P_{x^n}^k \otimes \mathcal{C}}\Big[U\Big(\hat{\mathbf{P}}\left(X_{0} | Y_{-k}^{k}\right)\Big)\Big] 
- \mathbb{E}_{P_{x^n}^k \otimes \mathcal{C}}\Big[U\Big(\textbf{P}\left(X_{0} | Y_{-k}^{k}\right)\Big)\Big]\Big| \nonumber \label{ineq:tri_1} \\
& + \Big|\mathbb{E}_{P_{x^n}^k \otimes \mathcal{C}}\Big[U\Big(\textbf{P}\left(X_{0} | Y_{-k}^{k}\right)\Big)\Big] - \mathbb{E}_{P_{x^n}^k \otimes \mathcal{C}}\Big[U\Big(\hat{\mathbf{P}}^{\delta}\left(X_{0} | Y_{-k}^{k}\right)\Big)\Big]\Big|\\
\leq \, & \Lambda_{\max}\cdot {\mathbb{E}} \|\textbf{P}(X_{0}|Y_{-k}^{k})-\hat{\mathbf{P}}(X_{0}|Y_{-k}^{k})\|_{1}
+ \Lambda_{\max}\cdot {\mathbb{E}} \|\textbf{P}(X_{0}|Y_{-k}^{k})-\hat{\mathbf{P}}^{\delta}(X_{0}|Y_{-k}^{k})\|_{1} \label{ineq:tri_2} \\
\leq \, & 2\Lambda_{\max}\cdot {\mathbb{E}} \|\textbf{P}(X_{0}|Y_{-k}^{k})-\hat{\mathbf{P}}(X_{0}|Y_{-k}^{k})\|_{1}
+ \Lambda_{\max}\cdot {\mathbb{E}} \|\hat{\textbf{P}}(X_{0}|Y_{-k}^{k})-\hat{\mathbf{P}}^{\delta}(X_{0}|Y_{-k}^{k})\|_{1} \label{ineq:tri_3}\\
\leq \, & 2\Lambda_{\max}\cdot {\mathbb{E}} \|\textbf{P}(X_{0}|Y_{-k}^{k})-\hat{\mathbf{P}}(X_{0}|Y_{-k}^{k})\|_{1}
+ \frac{\Lambda_{\max}\cdot M\cdot \delta}{2} \label{ineq:tri_4}\\
\leq \, &2\Lambda_{\max}\cdot \epsilon^*+ \frac{\Lambda_{\max}\cdot M\cdot \delta}{2},\label{ineq:tri_5}
    \end{align}
    in which (\ref{ineq:tri_1}) follows from the triangular inequality, (\ref{ineq:tri_2}) follows from (\ref{eqn_r}) and replacing $\hat{\mathbf{P}}$ with $\hat{\mathbf{P}}^{\delta}$ in (\ref{eq:e_def}), (\ref{ineq:tri_3}) follows from applying the triangular inequality once more, (\ref{ineq:tri_4}) follows from Lemma 3, and (\ref{ineq:tri_5}) follows from (\ref{eqn_s}). Note that probability vectors $\mathbf{P}, \hat{\mathbf{P}}$ in Lemma \ref{lemma_ex_est} replaced $\hat{\mathbf{P}}, \hat{\mathbf{P}}^{\delta}$ in Lemma \ref{lemma_ex_quan} respectively. \ \ \qed
    % By applying (\ref{eqn_r}), 
    % and then applying (\ref{eqn_r}) again after replacing $\hat{\mathbf{P}}$ in Lemma \ref{lemma_ex_est} with $\hat{\mathbf{P}}^{\delta}$ gives the second inequality. The fourth and last inequalities can be obtained from Lemma \ref{lemma_quan} and (\ref{eqn_s}) respectively. Note that probability vectors $\mathbf{P}, \hat{\mathbf{P}}$ in Lemma \ref{lemma_ex_est} replaced $\hat{\mathbf{P}}, \hat{\mathbf{P}}^{\delta}$ in Lemma \ref{lemma_ex_quan} respectively.
\end{proof}

% \begin{lemma}
% For every $n \geq 1$, $x^{n} \in \mathcal{A}^{n} $, measurable $g_k: \mathbb{R}^{2 k+1} \rightarrow \mathcal{A}$, and $\epsilon >0$, 
% \begin{align*}
%   \Pr\Big(\Big| \frac{1}{n-2k} \sum_{i=k+1}^{n-k} \Lambda\left(x_{i}, g_k\left(Y_{i-k}^{i+k}\right)\right)
% -\mathbb{E}_{P_{x^n}^{k} \otimes \mathcal{C}} [\Lambda\left(X_{0}, g_k\left(Y_{-k}^{k}\right)\right)] \Big|> \epsilon \Big)
% \leq 2(2k+1)\exp \left(-\frac{2\left(n-2k\right)}{\left(2k+1\right)} \epsilon^{2} \cdot \frac{1}{\Lambda_{\max }^{2}}\right).
% \end{align*}
% \end{lemma}
\begin{lemma}
For every $n \geq 1$, $x^{n} \in \mathcal{A}^{n} $, $\epsilon >0$ and measurable $g_k: \mathbb{R}^{2 k+1} \rightarrow \mathcal{A}$, 
\begin{equation*}
  \Pr\Big(\Big| \frac{1}{n-2k} \sum_{i=k+1}^{n-k} \Lb\Big(x_{i}, g_k\left(Y_{i-k}^{i+k}\right)\Big)
-\mathbb{E}_{P_{x^n}^{k} \otimes \mathcal{C}} \Big[\Lb\Big(X_{0}, g_k\left(Y_{-k}^{k}\right)\Big)\Big] \Big|> \epsilon \Big)
\leq 2(2k+1)\exp \left(-\frac{2\left(n-2k\right)}{\left(2k+1\right)} \epsilon^{2} \cdot \frac{1}{\Lambda_{\max }^{2}}\right).
\end{equation*}
\label{lemma_hoeffding_general}
\end{lemma}

\begin{proof}
We have the following:
    \begin{align}
      & \Pr\Big( \Big| \frac{1}{n-2k} \sum_{i=k+1}^{n-k} \Lb\Big(x_{i}, g_k\left(Y_{i-k}^{i+k}\right)\Big)
    -\mathbb{E}_{P_{x^n}^{k} \otimes \mathcal{C}} \Big[\Lb\Big(X_{0}, g_k\left(Y_{-k}^{k}\right)\Big)\Big] \Big|> \epsilon \Big)\nonumber\\
    = \, & 2\cdot \Pr\Big( \frac{1}{n-2k} \sum_{m=0}^{2k} \sum_{\substack{i \in\{k+1, \ldots, n-k\},\\ \lceil(i-m) /(2 k+1)\rceil =(i-m) /(2 k+1)}} 
    \Lb\Big(x_{i}, g_k\left(Y_{i-k}^{i+k}\right)\Big)
    -\mathbb{E}_{P_{x^n}^{k} \otimes \mathcal{C}} \Big[\Lb\Big(X_{0}, g_k\left(Y_{-k}^{k}\right)\Big)\Big] > \epsilon \Big)\nonumber\\
    \leq \, & 2(2k+1)\cdot \Pr\Big( \frac{2k+1}{n-2k}  \sum_{\substack{i \in\{k+1, \ldots, n-k\},\\ \lceil i /(2 k+1)\rceil =i /(2 k+1)}} 
    \Lb\Big(x_{i}, g_k\left(Y_{i-k}^{i+k}\right)\Big)
    -\mathbb{E}_{P_{x^n}^{k} \otimes \mathcal{C}} \Big[\Lb\Big(X_{0}, g_k\left(Y_{-k}^{k}\right)\Big)\Big] > \epsilon \Big)\label{lem6_1}\\
    \leq \, & 2(2k+1)\exp \left(-\frac{2\left(n-2k\right)}{\left(2k+1\right)} \epsilon^{2} \cdot \frac{1}{\Lambda_{\max }^{2}}\right).\label{lem6_2}
    \end{align}
 Note that if $|i-j|>2k$, $\Lb(x_{i}, g_k(Y_{i-k}^{i+k}))$ is independent from $\Lb(x_{j}, g_k(Y_{j-k}^{j+k}))$. (\ref{lem6_1}) follows  from the union bound, and (\ref{lem6_2}) follows from the fact that $\Lb(x_{i}, g_k(Y_{i-k}^{i+k}))- \mathbb{E}_{P_{x^n}^{k} \otimes \mathcal{C}} [\Lb(X_{0}, g_k(Y_{-k}^{k}))]$ is a zero-mean, bounded, independent random variable, and the Hoeffding's inequality.
%  gives the last inequality.
Thus, for every $k$th-order sliding window denoiser, difference between empirical loss and expected loss is vanishing with high probability.\ \qed
\end{proof}

\begin{lemma}
Let $\mathcal{F}_{\delta}^{k}$ denote the set of $\mathcal{A}^{2k+1}$-dimensional vecotrs with components in $[0,1]$ that are integer multiples of $\delta$. Note that $\hat{\mathbf{P}}^{\delta} \in \mathcal{F}_{\delta}^{k}$. Also, let $\mathcal{G}_{\delta}^{k}=\{\mathcal{B}(\mathbf{P})\}_{\mathbf{P} \in \mathcal{F}_{\delta}^{k}}$ be the class of $k$-th order sliding window denoiser defined by computing the Bayes response with respect to $\mathbf{P} \in \mathcal{F}_{\delta}^{k}$. Then, 
for every $n \geq 1$, $x^{n} \in \mathcal{A}^{n} $,  $\epsilon >0$ and $\mathcal{B}(\hat{\mathbf{P}}^{\delta}) \in \mathcal{G}_{\delta}^{k}$,
\begin{equation*}
\Pr\Big(\Big| L_{\hat{X}_{\text{NN}}^{\delta}}\left(x^{n}, Y^{n}\right)
-\mathbb{E}_{P_{x^n}^{k} \otimes \mathcal{C}} \Big[\Lb\Big(X_{0}, \mathcal{B}(\hat{\mathbf{P}}^{\delta}(X_0|Y_{-k}^{k}))\Big)\Big] \Big|> \epsilon \Big)
\leq \Big[\frac{1}{\delta}+1\Big]^{M}\cdot 2(2k+1)\exp \left(-\frac{2\left(n-2k\right)}{\left(2k+1\right)} \epsilon^{2} \cdot \frac{1}{\Lambda_{\max }^{2}}\right).
\end{equation*}
\label{lemma_NN_hoeffding}
\end{lemma}
\begin{proof}
We have
\begin{align}
    & \Pr\Big( \Big| L_{\hat{X}_{\text{NN}}^{\delta}}\left(x^{n}, Y^{n}\right)
    -\mathbb{E}_{P_{x^n}^{k} \otimes \mathcal{C}} \Big[\Lb\Big(X_{0}, \mathcal{B}(\hat{\mathbf{P}}^{\delta}(X_0|Y_{-k}^{k}))\Big)\Big] \Big|> \epsilon \Big)\nonumber\\
        = \, & \Pr\Big(\Big| \frac{1}{n-2k} \sum_{i=k+1}^{n-k} \Lb\Big(x_{i}, \mathcal{B}(\hat{\mathbf{P}}^{\delta}(X_i|Y_{i-k}^{i+k}))\Big)
        -\mathbb{E}_{P_{x^n}^{k} \otimes \mathcal{C}} \Big[\Lb\Big(X_{0}, \mathcal{B}(\hat{\mathbf{P}}^{\delta}(X_0|Y_{-k}^{k}))\Big)\Big] \Big|> \epsilon \Big)\nonumber\\
    \leq \, & \Pr\Big(\max_{g_k^* \in \mathcal{G}_{\delta}^{k}}\Big| \frac{1}{n-2k} \sum_{i=k+1}^{n-k} \Lb\Big(x_{i}, g_k^*\left(Y_{i-k}^{i+k}\right)\Big)
        -\mathbb{E}_{P_{x^n}^{k} \otimes \mathcal{C}} \Big[\Lb\Big(X_{0}, g_k^*\left(Y_{i-k}^{i+k}\right)\Big)\Big] \Big|> \epsilon \Big)\label{eq:lem7_1}\\
    \leq \, & \Big| \mathcal{G}_{\delta}^{k} \Big|\cdot 2(2k+1)\exp \left(-\frac{2\left(n-2k\right)}{\left(2k+1\right)} \epsilon^{2} \cdot \frac{1}{\Lambda_{\max }^{2}}\right)\label{eq:lem7_2}\\
    \leq \, &\Big[\frac{1}{\delta}+1\Big]^{M}\cdot 2(2k+1)\exp \left(-\frac{2\left(n-2k\right)}{\left(2k+1\right)} \epsilon^{2} \cdot \frac{1}{\Lambda_{\max }^{2}}\right),\label{eq:lem7_3}
    \end{align}
in which (\ref{eq:lem7_1}) follows from considering the uniform convergence, (\ref{eq:lem7_2}) follows from the union bound, and (\ref{eq:lem7_3}) follows from the crude upper bound on the cardinality $| \mathcal{G}_{\delta}^{k}|$.
% Lemma \ref{lemma_hoeffding_general} and the union bound give the second inequality.
Note that the window size $k$ in the superscript of upper bound for the cardinality ($[\frac{1}{\delta}+1]^{M}$) is removed compared to that of Gen-DUDE ($ [\frac{1}{\delta}+1]^{M^{2k+1}}$).
The distinction between them follows from difference in modeling where Gen-CUDE tries to directly model the marginal posterior distribution with neural network rather than the joint posterior of $(2k+1)$-tuple. \ \qed
\end{proof}

\setcounter{theorem}{0}

% \begin{theorem}
% For all k, $n \geq 1$, $\delta>0$, $\epsilon >0$, $\epsilon^*>0$, $M>0$,
% \begin{align*}
%     \Pr\Big( & \Big| L_{\hat{X}_{\text{NN}}^{\delta}}\left(x^{n}, Y^{n}\right)
% - D_{x^n}^k \Big|> \epsilon + \Lambda_{\max}\cdot \Big( 3\epsilon^* + \frac{ M \cdot \delta}{2}\Big) \Big)\\
% & \leq \Big[\frac{1}{\delta}+1\Big]^{M}\cdot 2(2k+1)\exp \left(-\frac{2\left(n-2k\right)}{\left(2k+1\right)} \epsilon^{2} \cdot \frac{1}{\Lambda_{\max }^{2}}\right).
% \end{align*}
% \end{theorem}

Now, we prove our main theorm.
\begin{theorem}
\label{theorem_1}
Consider $\epsilon^*$ in Lemma \ref{lemma_assumption2_to_15}. Then, 
for all $k,n\geq 1$, $\delta>0$, and $\epsilon > \Lambda_{\max}\cdot ( 3\epsilon^* + \frac{M \cdot \delta}{2})$, and for all $x^n$,
\begin{align*}
    &\Pr\Big(| L_{\hat{X}_{\text{NN}}^{n,\delta}}\left(x^{n}, Y^{n}\right)
- D_{x^n}^k |> \epsilon  \Big)
\leq C_1(k,\delta,M)
% 2(2k+1)\cdot\Big[\frac{1}{\delta}+1\Big]^{M} \\
\exp \Big(-\frac{2\left(n-2k\right)}{\left(2k+1\right)} C_2(\epsilon, \epsilon^*, \Lambda_{\max}, M, \delta) \Big),\nonumber
\end{align*}
in which $C_1(k,\delta,M)\triangleq2(2k+1)[\frac{1}{\delta}+1]^{M}$ and 
% $ \mathcal{D} $ is defined as 
$C_2(\epsilon, \epsilon^*, \Lambda_{\max}, M, \delta)\triangleq ( \epsilon-\Lambda_{\max}\cdot ( 3\epsilon^* + \frac{M \cdot \delta}{2}) )^2 \cdot \frac{1}{\Lambda_{\max }^{2}}$.

\end{theorem}
\paragraph{Proof of theorem 1:} We utilize all the Lemmas given above to prove the theorem. We have
\begin{align}
    & \Pr\Big( \Big| L_{\hat{X}_{\text{NN}}^{\delta}}\left(x^{n}, Y^{n}\right)
        - D_{x^n}^k \Big|> \epsilon  \Big)\nonumber\\
        = & \Pr\Big(\Big| \frac{1}{n-2k} \sum_{i=k+1}^{n-k} \Lb\Big(x_{i}, \mathcal{B}(\hat{\mathbf{P}}^{\delta}(X_i|Y_{i-k}^{i+k}))\Big)
        -\mathbb{E}_{P_{x^n}^{k} \otimes \mathcal{C}} \Big[\Lb\Big(X_{0}, \mathcal{B}(\mathbf{P}(X_0|Y_{-k}^{k}))\Big)\Big] \Big|
        > \epsilon  \Big) \label{eq:thm1-1} \\
    \leq \, & \Pr\Big(\Big| \frac{1}{n-2k} \sum_{i=k+1}^{n-k} \Lb\Big(x_{i}, \mathcal{B}(\hat{\mathbf{P}}^{\delta}(X_i|Y_{i-k}^{i+k}))\Big)
        -\mathbb{E}_{P_{x^n}^{k} \otimes \mathcal{C}} \Big[\Lb\Big(X_{0}, \mathcal{B}(\hat{\mathbf{P}}^{\delta}(X_0|Y_{-k}^{k}))\Big)\Big] \Big|\nonumber\\
        & +  \Big| \mathbb{E}_{P_{x^n}^{k} \otimes \mathcal{C}} \Big[\Lb\Big(X_{0}, \mathcal{B}(\hat{\mathbf{P}}^{\delta}(X_0|Y_{-k}^{k}))\Big)\Big]
        -\mathbb{E}_{P_{x^n}^{k} \otimes \mathcal{C}} \Big[\Lb\Big(X_{0}, \mathcal{B}(\hat{\mathbf{P}}(X_0|Y_{-k}^{k}))\Big)\Big] \Big| \nonumber \\
        & + \Big| \mathbb{E}_{P_{x^n}^{k} \otimes \mathcal{C}} \Big[\Lb\Big(X_{0}, \mathcal{B}(\hat{\mathbf{P}}(X_0|Y_{-k}^{k}))\Big)\Big] 
        -\mathbb{E}_{P_{x^n}^{k} \otimes \mathcal{C}} \Big[\Lb\Big(X_{0}, \mathcal{B}(\mathbf{P}(X_0|Y_{-k}^{k}))\Big)\Big] \Big|
        >\epsilon \Big) \label{ineq:thm1-1}\\ 
        \leq \,  & \Pr\Big(\Big| \frac{1}{n-2k} \sum_{i=k+1}^{n-k} \Lb\Big(x_{i}, \mathcal{B}(\hat{\mathbf{P}}^{\delta}(X_i|Y_{i-k}^{i+k}))\Big)
        -\mathbb{E}_{P_{x^n}^{k} \otimes \mathcal{C}} \Big[\Lb\Big(X_{0}, \mathcal{B}(\hat{\mathbf{P}}^{\delta}(X_0|Y_{-k}^{k}))\Big)\Big] \Big| >  \epsilon-\Lambda_{\max}\cdot ( 3\epsilon^* + \frac{M \cdot \delta}{2}) \Big) \label{ineq:thm1-2}\\
    \leq \, & \Big[\frac{1}{\delta}+1\Big]^{M}\cdot 2(2k+1)\exp \left(-\frac{2\left(n-2k\right)}{\left(2k+1\right)} \cdot \Big( \epsilon-\Lambda_{\max}\cdot ( 3\epsilon^* + \frac{M \cdot \delta}{2})\Big)^{2} \cdot \frac{1}{\Lambda_{\max }^{2}}\right)\label{ineq:thm1-3})\\
    = \, & C_1(k,\delta,M) \exp \Big(-\frac{2\left(n-2k\right)}{\left(2k+1\right)} C_2(\epsilon, \epsilon^*, \Lambda_{\max}, M, \delta) \Big),
    \end{align}
where (\ref{eq:thm1-1}) follows from the definition of $ L_{\hat{X}_{\text{NN}}^{\delta}}(x^{n}, Y^{n})$ and $D_{x^n}^k$, (\ref{ineq:thm1-1}) follows from triangle inequality, (\ref{ineq:thm1-2}) follows from applying Lemma \ref{lemma_ex_quan} and Lemma \ref{lemma_ex_est}, and (\ref{ineq:thm1-3}) follows from Lemma \ref{lemma_NN_hoeffding}. Thus, we proved the theorem. \ \qed

% The triangle inequality tells (\ref{ineq:thm1-1}). By applying Lemma \ref{lemma_ex_quan} and Lemma \ref{lemma_ex_est}, we can obtain (\ref{ineq:thm1-2}). Lemma \ref{lemma_NN_hoeffding} gives (\ref{ineq:thm1-3}).

\clearpage
\section{Noise Channel Densities}\label{appendix:channelgraph}
%%%%%%%%%%%%%%%%%%%%%%%%%%%%%%%%%%%%%%%%%%%%%%%%%%%%%%%%%%%%%%%%%%%%%%%%%%%%%%
Here, we show the noisy channel density $\{f_x(y)\}_{x\in\mathcal{O}}$ used for the experiments in Section 5.2 and Section 5.3 of the paper. Figure \ref{prob_dist_syn} shows the channel densities for the synthetic data experiments in Section 5.2, and Figure \ref{prob_dist_DNA} shows the channel densities for the 454 and Ion Torrent data experiments in Section 5.3. 

\begin{figure*}[!ht]
\centering
\subfigure[][$|\mathcal{A}|=2$]{
\includegraphics[width=0.40\columnwidth]{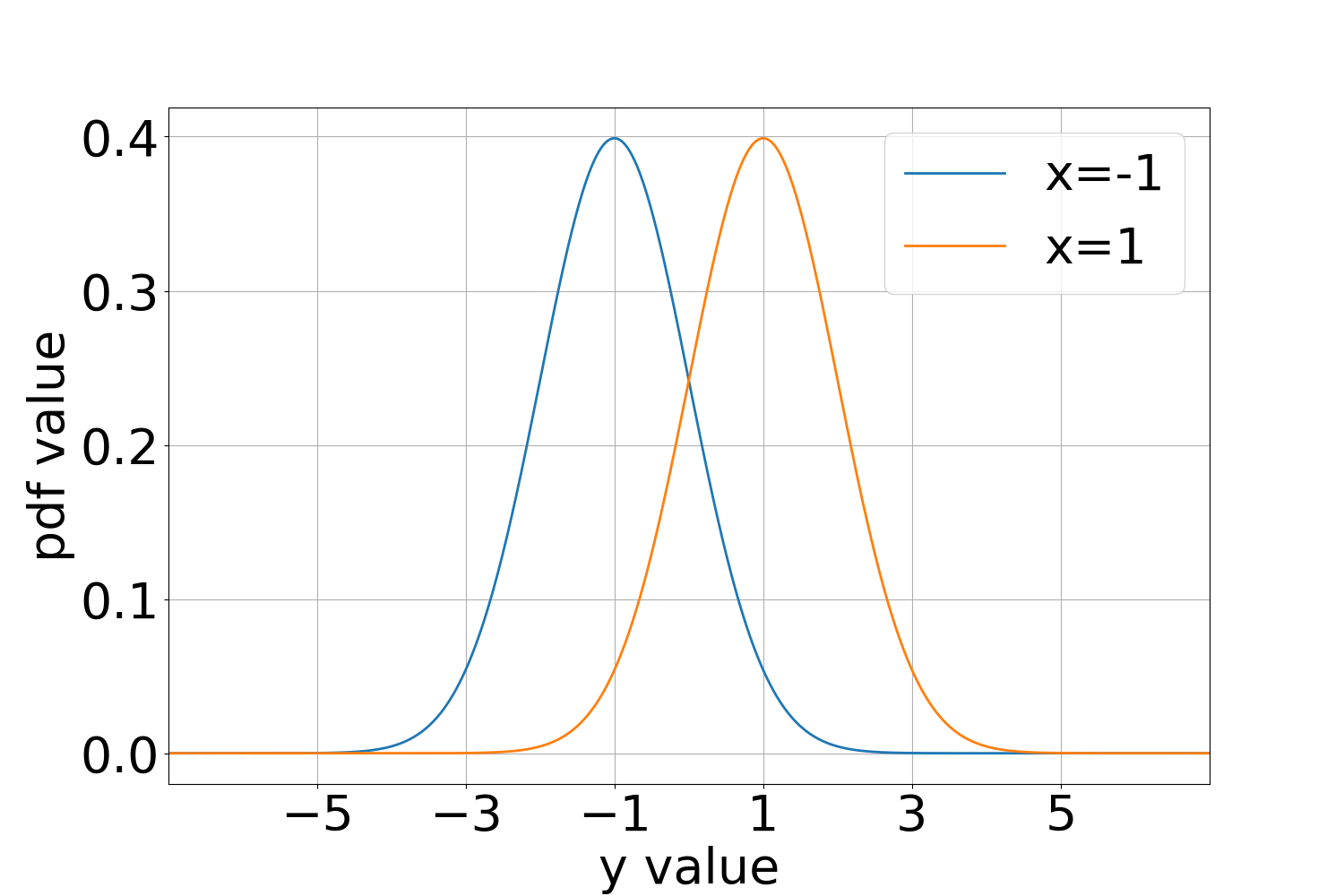}
\label{binary_pdf}
}
\subfigure[][$|\mathcal{A}|=4$]{
\includegraphics[width=0.40\columnwidth]{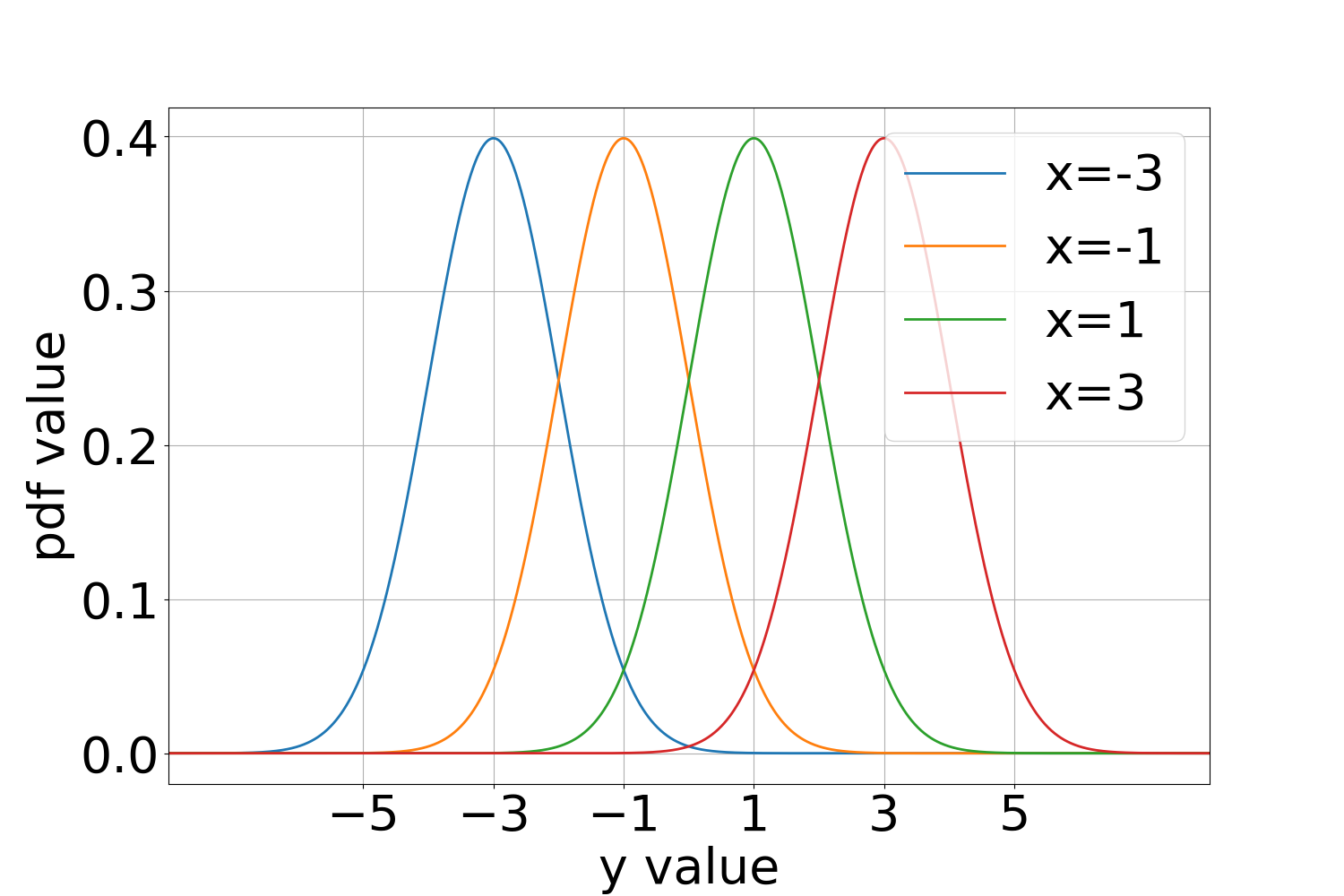}
\label{tetra_pdf}
}
\subfigure[][$|\mathcal{A}|=10$]{
\includegraphics[width=0.50\columnwidth]{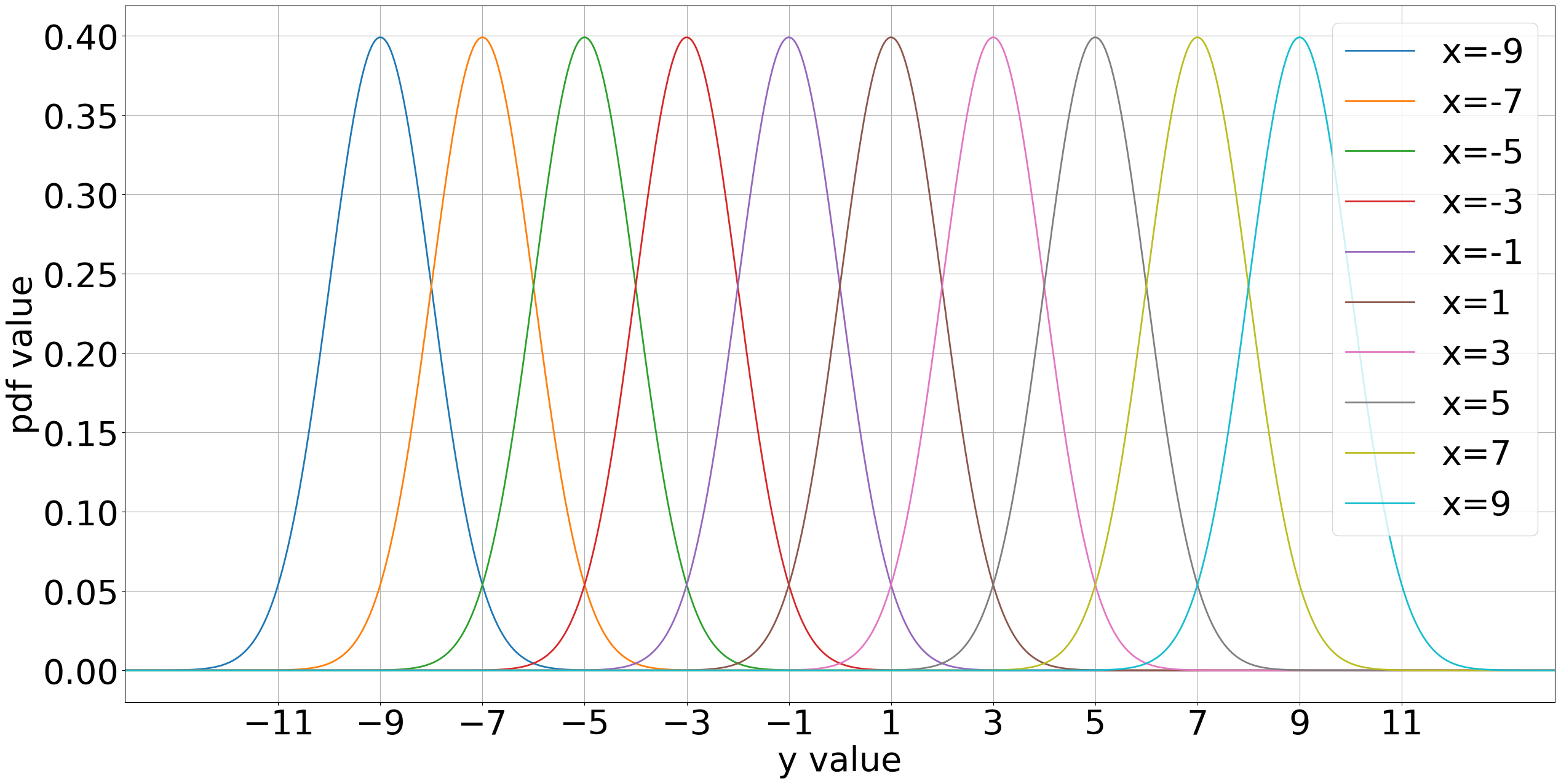}
\label{deca_pdf}
}
\caption{\label{prob_dist_syn}
Noisy channel densities used for the synthetic data experiments. 
}
\end{figure*}
%%%%%%%%%%%%%%%%%%%%%%%%%%%%%%%%%%%%%%%%%%%%%%%%%%%%%%%%%%%%%%%%%%%%%%%%%%%%%%

%%%%%%%%%%%%%%%%%%%%%%%%%%%%%%%%%%%%%%%%%%%%%%%%%%%%%%%%%%%%%%%%%%%%%%%%%%%%%%
\begin{figure*}[!ht]
\centering
\subfigure[][454 Pyrosequencing]{
\includegraphics[width=0.45\columnwidth]{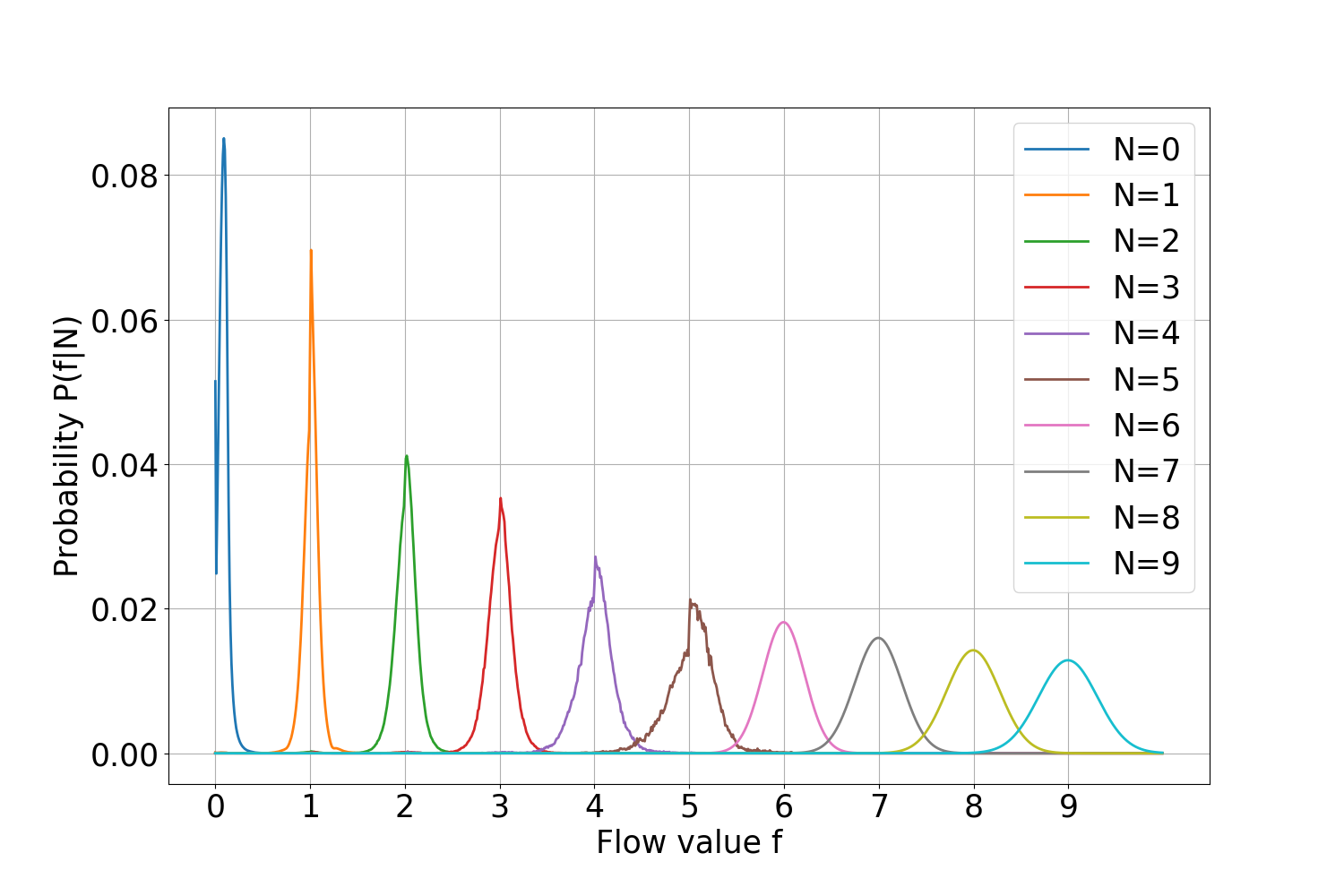}
\label{454_pdf}
}
\subfigure[][Ion Torrent]{
\includegraphics[width=0.45\columnwidth]{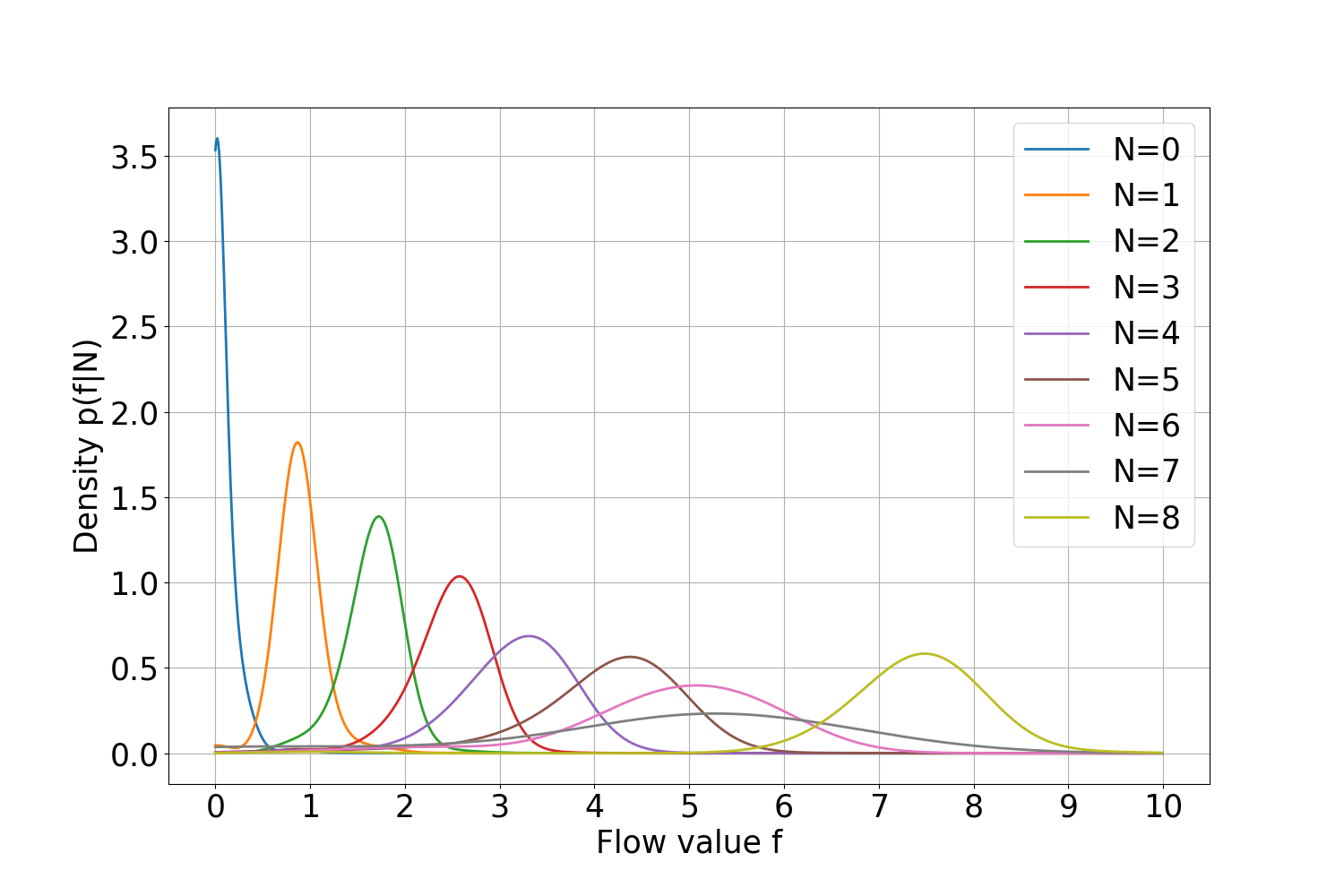}
\label{ion_pdf}
}
\caption{\label{prob_dist_DNA}
Probability ensities of the flowgram-values for the homopolymer lengths in each DNA sequencer.  For Ion Torrent, we estimated channel density using Gaussian kernel density estimation with bandwidth=0.6 on the separated holdout dataset.
}
\end{figure*}
%%%%%%%%%%%%%%%%%%%%%%%%%%%%%%%%%%%%%%%%%%%%%%%%%%%%%%%%%%%%%%%%%%%%%%%%%%%%%%

%\clearpage
\section{Normalized Error Rate Graph for DNA Experiments}\label{appendix:biterrorrate}
%%%%%%%%%%%%%%%%%%%%%%%%%%%%%%%%%%%%%%%%%%%%%%%%%%%%%%%%%%%%%%%%%%%%%%%%%%%%%%
Figure \ref{Error_DNA} shows the denoising performance measured by the Hamming loss. Note the similarity score in the paper is computed after converting the integer-valued denoised sequence (homopolymer length) back to a DNA sequence. We observe the error patterns are similar to those in Figure 2 of the paper. 
\begin{figure*}[!ht]
\centering
\subfigure[][454 Pyrosequencing]{
\includegraphics[width=0.45\columnwidth]{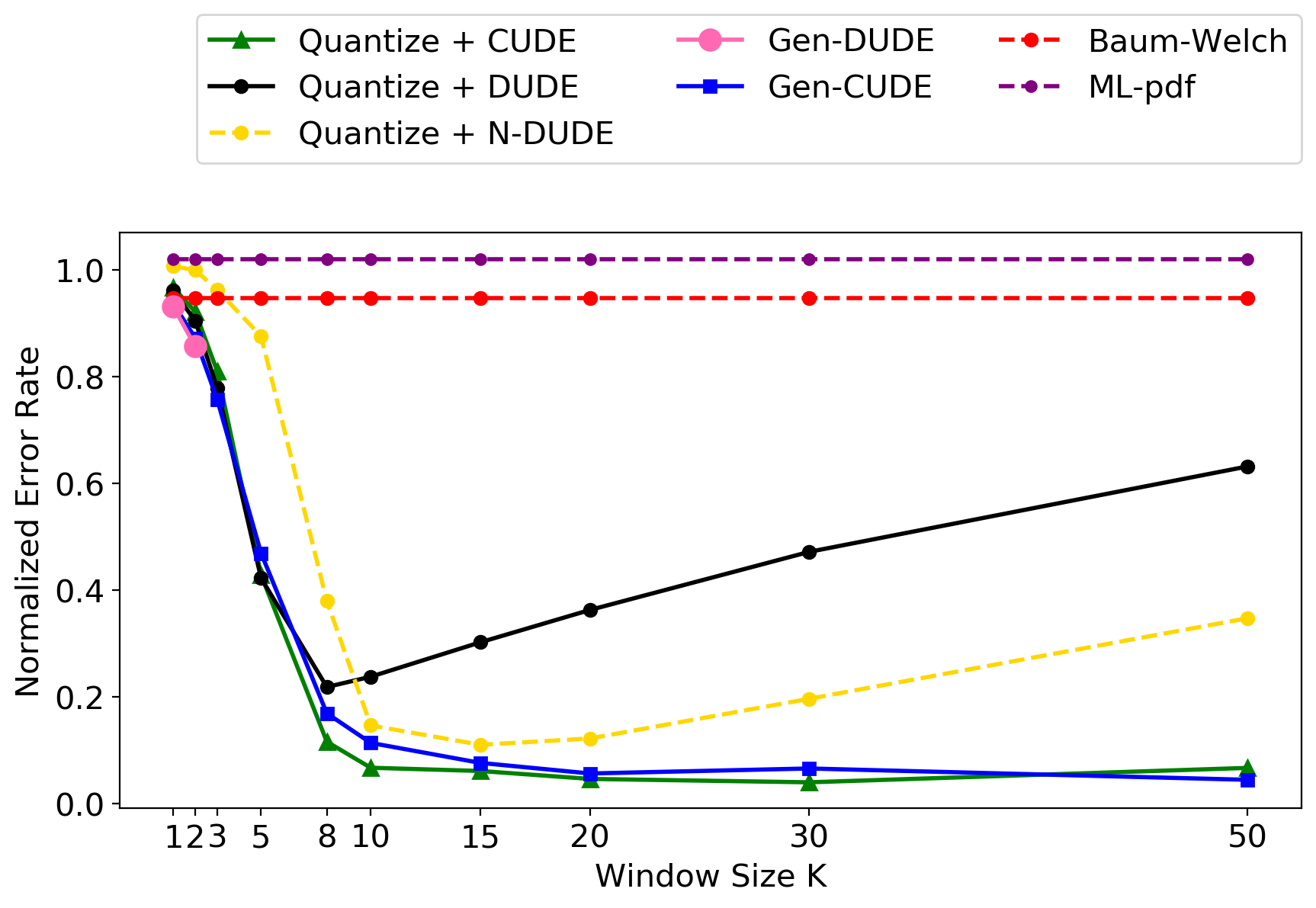}

}
\subfigure[][Ion Torrent]{
\includegraphics[width=0.45\columnwidth]{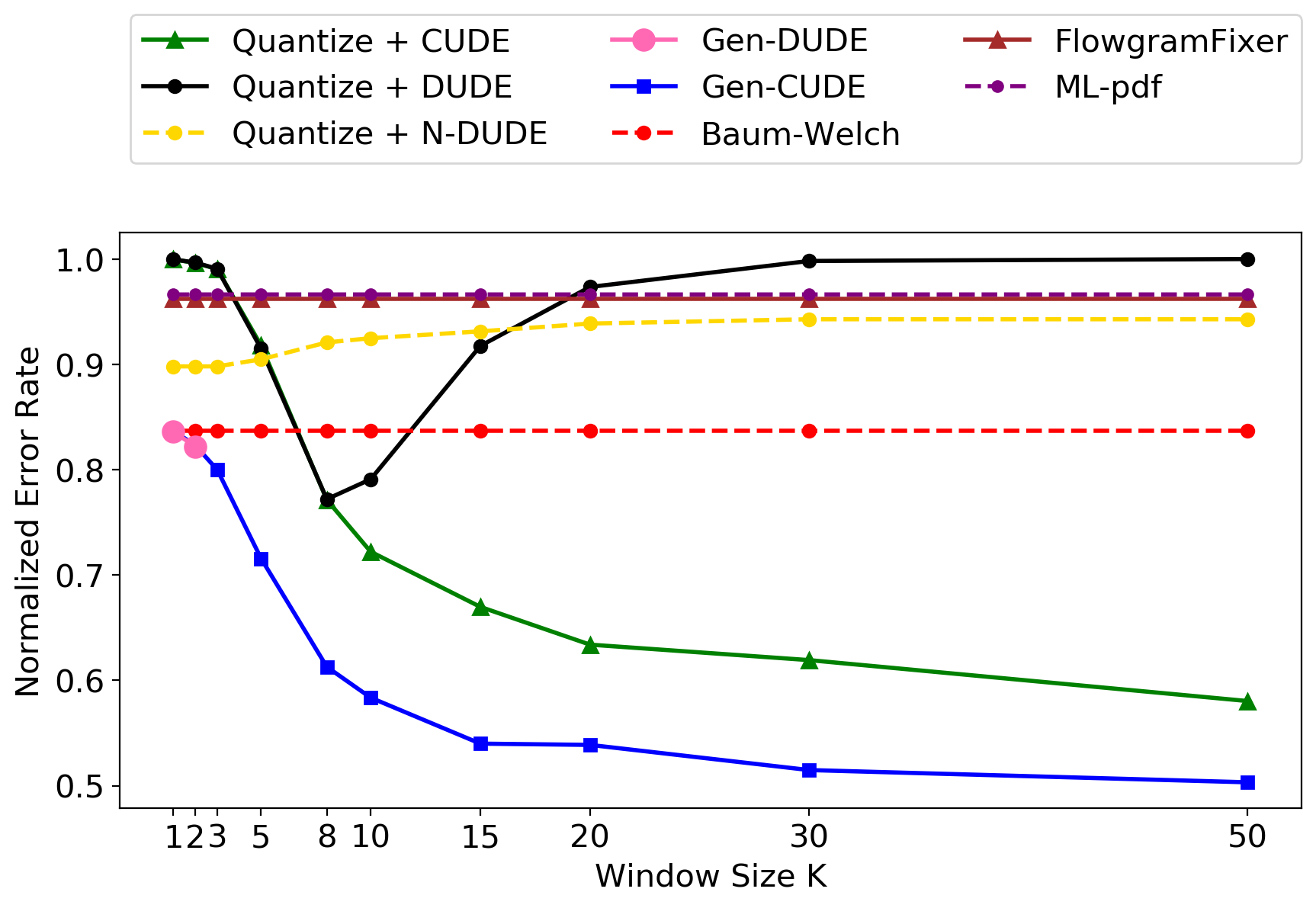}

}
\caption{\label{Error_DNA}
Normalized error rate for DNA source data.
}
\end{figure*}
%%%%%%%%%%%%%%%%%%%%%%%%%%%%%%%%%%%%%%%%%%%%%%%%%%%%%%%%%%%%%%%%%%%%%%%%%%%%%%

%

%%%%%%%%%%%%%%%%%%%%%%%%%%%%%%%%%%%%%%%%%%%%%%%%%%%%%%%%%%%%%%%%%%%%%%%%%%%%%%
\clearpage
\section{Error Rate Graph for Randomized Quantizers}\label{appendix:RandomizedQuantizer}

We note that the quantizer $Q(\cdot)$ can be freely selected for Gen-CUDE as long as the induced DMC, $\Pi$, is invertible. To show the small effect of the quantizer to the final denoising performance, we designed two additional experiments for the $|\mathcal{A}|=4$ case of Figure 1(a) in the paper. As described in the first paragraph of Section 5.2, the source symbol was encoded as $\{+3,+1,-1,-3\}$ and the decision boundaries of the original $Q(\cdot)$ was $\{-2,0,+2\}$. 

\begin{figure*}[h]
\centering

\subfigure[][Normalized Error Rate]{
\includegraphics[width=0.95\columnwidth]{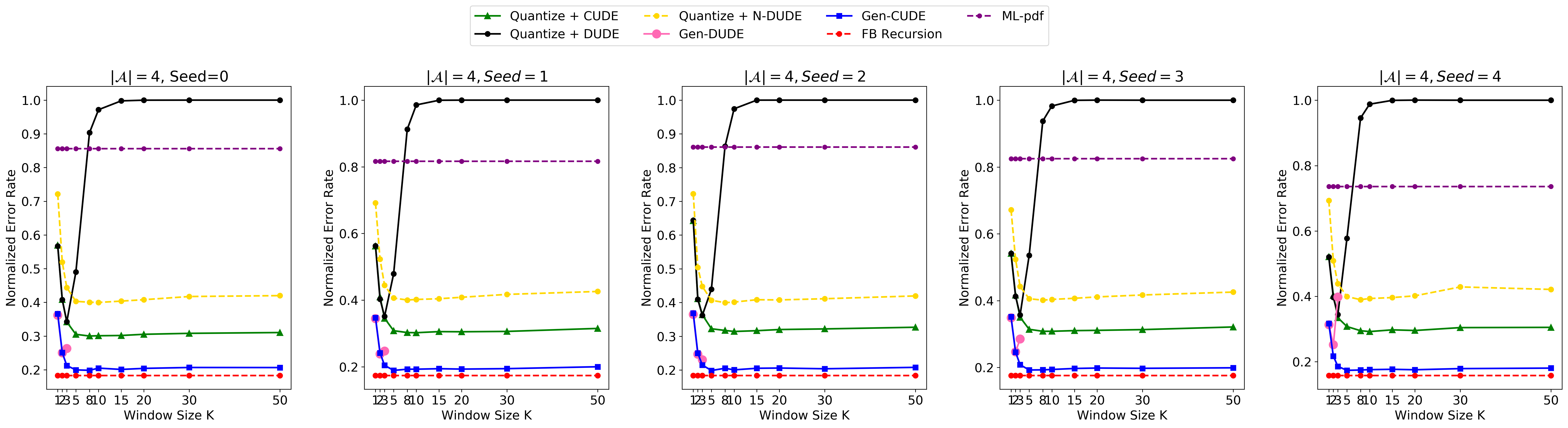}
\label{Each_Error_Q}
}
%\caption{\label{fig1)}}

\subfigure[][Average Error Rate]{
\includegraphics[width=0.90\columnwidth]{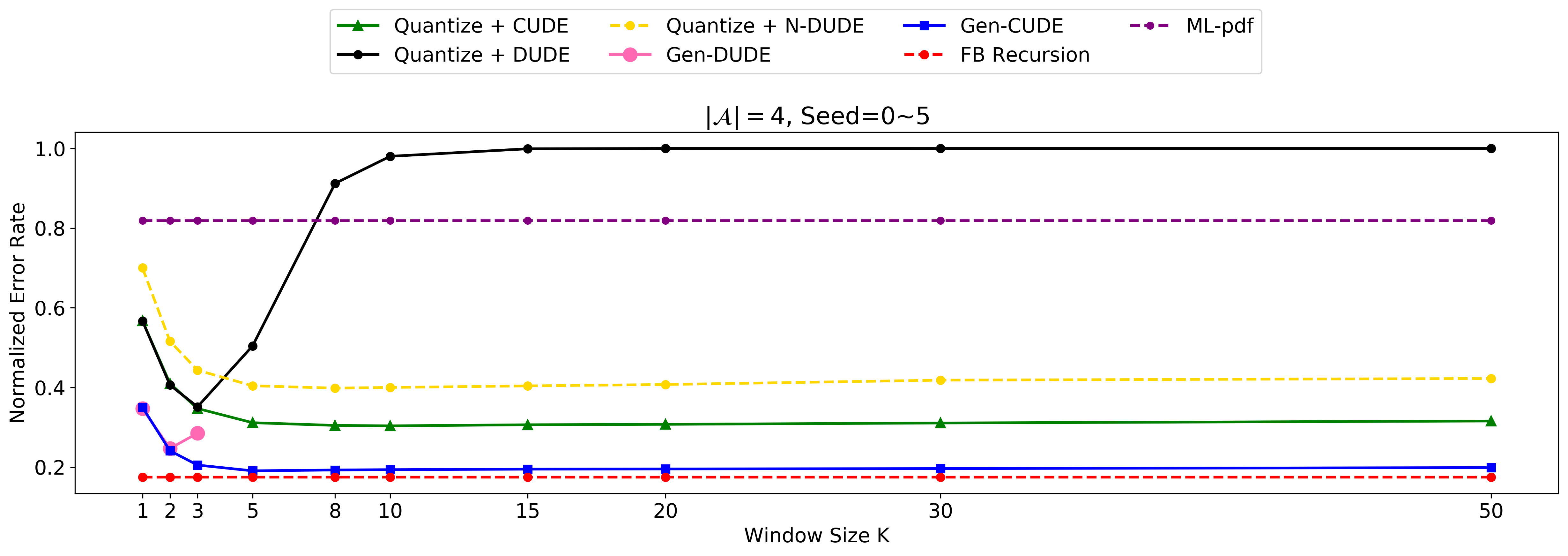}
\label{Avg_Error_Q}
}
\caption{Error Rate for Five Randomized Quantizers}
\end{figure*}

In Figure \ref{Each_Error_Q}, we show the results of using five randomized quantizers, of which decision boundaries were obtained by uniform sampling from the intervals, $[-3,-1], [-1,+1], [+1,+3]$, respectively. The 5 different resulting quantizers’ decision boundaries were the following:
\begin{itemize}
    \item Seed 0 : $[-1.59, 0.73, 2.09]$
    \item Seed 1 : $[-1.18, 0.27, 2.46]$
    \item Seed 2 : $[-2.29, -0.59, 2.49]$
    \item Seed 3 : $[-1.96, -0.41, 1.12]$
    \item Seed 4 : $[-1.13, 0.81, 1.61]$.
\end{itemize}

The five figures in Figure \ref{Each_Error_Q} show the performance for each quantizer, and Figure \ref{Avg_Error_Q} shows the average error rate of them, which looks quite similar the one shown in Figure 1(a). We can clearly observe that the different quantizers have little effect in the final denoising performance for \texttt{Gen-CUDE}.  In contrast, we observe that \texttt{Gen-DUDE} or \texttt{Quantize+DUDE} have more sensitivity to the choice of the quantizer.

Furthermore, we note that our Gen-CUDE does not require to have the same number of the quantized symbols as the input symbols, either. In such cases, the $\mathbf{\Pi}^{-1}$ can be simply replaced with a pseudo-inverse as long as $\mathbf{\Pi}$ has full row-rank. Figure \ref{Avg_Error_nonsquare_Q} is the result of averaging the performances of using five randomized quantizers, of which decision boundaries are randomly selected from the intervals $[-2.7, -2.3]$, $[-1.7, -1.3]$, $[-0.7, -0.3]$, $[0.3, 0.7]$, $[1.3, 1.7]$, $[2.3, 2.7]$, respectively. (Thus, $Q(\cdot)$ has 7 regions.) The used boundaries are as following:

\begin{figure*}[t]
\centering\includegraphics[width=0.90\columnwidth]{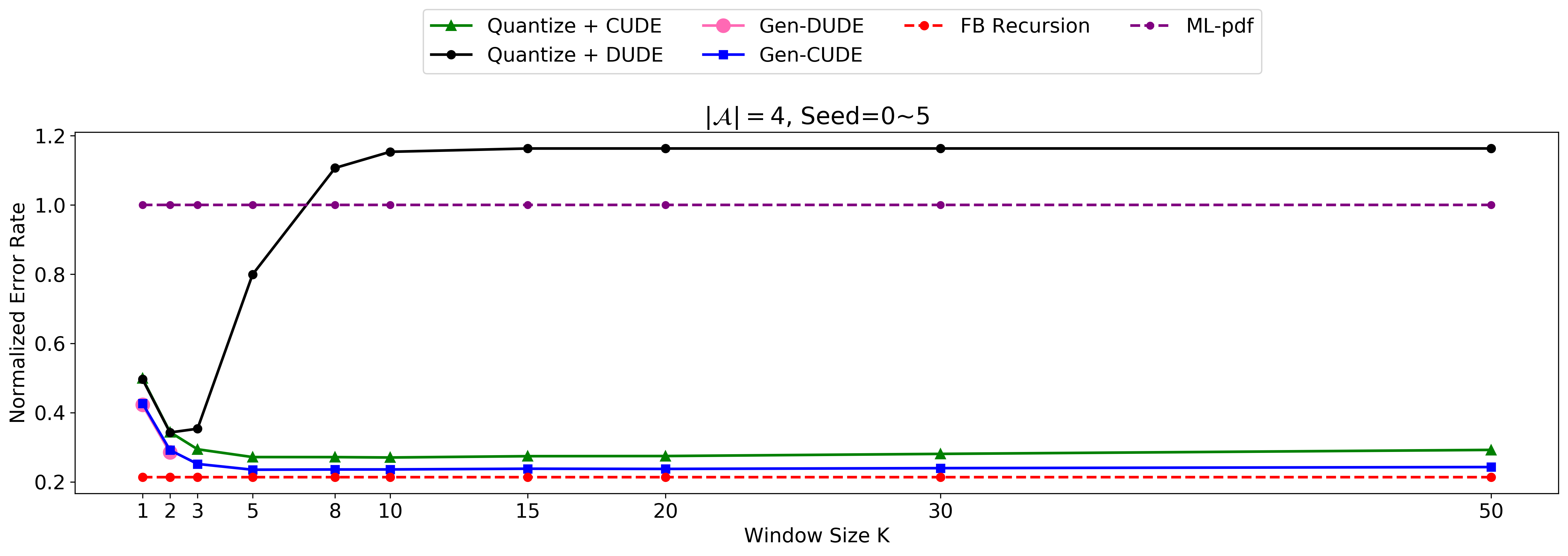}
\caption{Average Error Rate for Non-square Channel Matrix Case \label{Avg_Error_nonsquare_Q}}
\end{figure*}

\begin{itemize}
    \item Seed 0 : $[-2.42, -1.35, -0.48, 0.57, 1.44, 2.36]$
    \item Seed 1 : $[-2.34, -1.45, -0.41, 0.55, 1.55, 2.32]$
    \item Seed 2 : $[-2.56, -1.62, -0.4, 0.59, 1.56, 2.55]$
    \item Seed 3 : $[-2.49, -1.58, -0.68, 0.59, 1.49, 2.51]$
    \item Seed 4 : $[-2.33, -1.34, -0.58, 0.5, 1.67, 2.64].$
\end{itemize}

Again, we see little difference in the performance for \texttt{Gen-CUDE} compared to Figure \ref{Avg_Error_nonsquare_Q} and Figure 1(a) ($|\mathcal{A}|=4$ case) in the manuscript.

\end{document}